\documentclass[11pt]{article}

\usepackage[paper=letterpaper, margin=1in]{geometry}
\usepackage[T1]{fontenc}
\usepackage[utf8]{inputenc}
\usepackage[tbtags]{amsmath}
\allowdisplaybreaks
\usepackage{amssymb,amsthm,mathtools, bbm}
\usepackage{thmtools, thm-restate}
\usepackage{physics}
\usepackage{hyperref}
\usepackage[svgnames]{xcolor}
\usepackage[capitalise,nameinlink]{cleveref}
\definecolor{links}{RGB}{20, 50, 255}
\definecolor{cites}{RGB}{0, 200, 0}
\definecolor{urls}{RGB}{255, 116, 0}
\hypersetup{colorlinks=true,citecolor=blue,linkcolor=purple,urlcolor=blue}
\definecolor{MyPurple}{cmyk}{0.45,0.86,0,0}
\renewcommand{\emph}[1]{{\color{MyPurple}{\em #1}}}
\usepackage[numbers,compress]{natbib}
\usepackage{doi}
\usepackage{tikz} 
\usepackage{pgfplots}
\pgfplotsset{compat=1.14}
\usepgfplotslibrary{fillbetween}
\usepackage{hyphenat}
\usepackage[font=footnotesize]{caption}
\usepackage{subcaption}
\usepackage{appendix}
\usepackage{nicefrac}
\usepackage{comment}
\usepackage{xspace}
\usepackage{lineno}
\linenumbers

\usepackage{algorithm}
\usepackage{algorithmic}

\makeatletter
\newcommand{\functionHeader}[2]{%
  \item[]\textbf{Function #1:} \hfill #2\\[0.5ex]%
  \setcounter{ALC@line}{0}%
}
\makeatother

\usepackage{enumitem}
\usepackage{palatino}
\usepackage[bibliography=common, appendix=inline]{apxproof}

\newtheorem{question}{Question}
\declaretheorem[numberwithin=section]{theorem}
\declaretheorem[sibling=theorem]{lemma}
\newtheorem{proposition}[theorem]{Proposition}
\newtheorem{corollary}[theorem]{Corollary}
\newtheorem{claim}[theorem]{Claim}
\newtheorem{remark}[theorem]{Remark}
\newtheorem{observation}[theorem]{Observation}

\newtheorem{definition}[theorem]{Definition}

\newtheorem*{theorem*}{Theorem}

\usepackage{algorithm}
\usepackage{algorithmic}

\usepackage{enumitem}

\newcommand{\capd}{\bar{d}}
\renewcommand{\norm}[1][\cdot]{\left\lVert#1\right\rVert}
\DeclareMathOperator{\DP}{DP}

\newcommand{\OPT}{\ensuremath{\mathtt{OPT}}\xspace}
\newcommand{\children}{\ensuremath{\mathtt{children}}\xspace}

\DeclareMathOperator{\bdelay}{delay\xspace}
\DeclareMathOperator{\bserve}{serve\xspace}
\newcommand{\greedytau}{\ensuremath{\mathrm{Greedy}_{\tau}}\xspace}
\newcommand{\nextack}{\ensuremath{\mathrm{next}}\xspace}

\newcommand{\A}{\mathcal{A}}

\newcommand{\B}{\mathcal{B}}
\newcommand{\BA}{\B_{\A}}

\newcommand{\OPTTCP}{\OPT_{\mathrm{TCP}}}
\newcommand{\OPTPP}{\OPT_{\mathrm{PP}}}

\usepackage{multicol}

\begin{document}
\begin{titlepage}
\nolinenumbers
\title{Online TCP Acknowledgment under General Delays}

\date{}

\author{Sujoy Bhore\thanks{Department of Computer Science \& Engineering, Indian Institute of Technology Bombay, India. Work supported in part by ANRF ARG-MATRICS, Grant 002465. Email: \href{sujoy@cse.iitb.ac.in}{sujoy@cse.iitb.ac.in}}\and Michał Pawłowski\thanks{Faculty of Mathematics, Informatics and Mechanics, University of Warsaw, Department of Computer, Control and Management Engineering, Sapienza University of Rome and IDEAS NCBR. Supported by the Polish National Science Center grant no 2024/53/N/ST6/04119. Email: \href{michal.pawlowski@mimuw.edu.pl}{michal.pawlowski196@gmail.com}} \and Seeun William Umboh\thanks{School of Computing and Information Systems,
The University of Melbourne and ARC Training Centre in Optimisation Technologies, Integrated Methodologies, and Applications (OPTIMA). Supported by the Australian Government through the Australian Research Council (ARC) DP240101353. Email \href{william.umboh@unimelb.edu.au}{william.umboh@unimelb.edu.au}}}

\def\thepage{}
\thispagestyle{empty}
\maketitle

\setcounter{page}{0}
\begin{abstract}

In a seminal work, Dooly, Goldman, and Scott (STOC~1998; JACM~2001) introduced the classic \emph{Online TCP Acknowledgment} problem. In this problem, a sequence of $n$ packets arrives over time, and the objective is to minimize both the number of acknowledgments sent and the total delay experienced by the packets. They showed that a natural greedy algorithm, which acknowledges when the delay of pending packets equals the acknowledgment cost, is $2$-competitive.

Online TCP Acknowledgment is the canonical online problem with delay, capturing the fundamental tradeoff between reducing service cost through batching and the delay incurred by pending requests. Prior work has largely focused on richer service-cost models, e.g., Joint Replenishment and Multi-Level Aggregation. However, other than the work of Albers and Bals (SODA~2003), which studies maximum delay and closely related objectives, not much is known about general delay costs beyond the sum of delay costs of requests.

In this work, we study Online TCP Acknowledgment under two generalized delay-cost models that we call \emph{batch-aware} and \emph{batch-oblivious}. In the batch-aware model, each batch incurs a delay cost that depends on the packet delays within that batch. For the max-over-batches objective, which generalizes Albers and Bals, we show that greedy remains $2$-competitive for every monotone batch-delay function. For the sum-over-batches objective, the picture changes sharply: greedy is $\Omega(n)$-competitive, and the optimal deterministic competitive ratio is $\Theta(\log n)$. Our matching upper bound requires only the minimal assumption that the batch delay function is monotone.

In the batch-oblivious model, the delay cost is a function of the global packet-delay vector. We show that greedy is $2$-competitive for continuous submodular delay costs, and more generally under a weaker zero-coordinate diminishing-marginals condition. This yields $2$-competitive algorithms for $\ell_p$ norms, $\textsf{Top-}k$ norms, and ordered norms. Using the submodular-norm approximation of Patton, Russo, and Singla, we also obtain an $O(\log n)$-competitive algorithm for arbitrary symmetric norms.
\end{abstract}

\end{titlepage}
\nolinenumbers

\newpage	
\setcounter{page}{1}

\section{Introduction}

The \emph{Online TCP Acknowledgment} problem, introduced by Dooly, Goldman, and Scott in their seminal STOC 1998 paper~\cite{dooly1998tcp}, arose from critical challenges in computer networking. Specifically, it addresses the dynamic management of acknowledgment delays within the Transmission Control Protocol (TCP), a cornerstone of the Internet's architecture (see~\cite{allman1998generation, clark1982rfc0813, de2005dynamic}). The TCP protocol requires the recipient of a stream of packets to acknowledge every packet. On the other hand, excessive delay can disrupt TCP's congestion control mechanisms. Consequently, limiting the latency between a packet's arrival and its acknowledgment is vital for maintaining network performance and protocol stability. %

In the Online TCP Acknowledgment problem---the TCP problem henceforth---a sequence of packets arrive online over time. At each time, we can transmit an acknowledgment to acknowledge all currently pending packets at a cost of $1$. Each packet accrues a delay cost until it is acknowledged. More formally, the TCP problem involves a sequence of packet arrival times \( A = ( a_1, \ldots, a_n ) \), and the objective is to partition \( A \) into \( k \) contiguous subsequences \( \sigma_1, \sigma_2, \ldots, \sigma_k \), for some $k$. The end of each subsequence corresponds to the time an acknowledgment is transmitted. Thus, each \( \sigma_i \) is a \emph{batch} of packets that are acknowledged together by the acknowledgment sent at time $t_i$. The batches $\sigma_i$ and the acknowledgment times $t_i$ are required to satisfy the following constraints: each packet can only be acknowledged after its arrival, i.e.~$t_i \geq a_j$ for every $j \in \sigma_i$; and every packet must be acknowledged, i.e.~the final acknowledgment time $t_k \geq a_n$, the arrival time of the last packet. The tradeoff between the cost of acknowledgments and the cost of delaying acknowledgments is modeled with the objective function
\begin{equation}
  \label{eq:obj}
  k + \sum_{i=1}^k \bdelay(\sigma_i, t_i).
\end{equation}
The number of acknowledgments $k$ is called the \emph{acknowledgment cost} and $\sum_{i=1}^k \bdelay(\sigma_i, t_i)$ the \emph{delay cost}. We refer to $\bdelay(\cdot, \cdot)$ as the \emph{batch delay cost function}.

\paragraph{TCP with Additive Delay Cost.}
The most classic and well-studied setting is when the delay function $\bdelay(\sigma_i,t_i) = \sum_{j \in \sigma_i} (t_i - a_j)$, i.e.~the delay cost is simply the total waiting time between each packet's arrival time and its acknowledgment time. The competitiveness of this problem has been completely characterized. For deterministic algorithms, Dooly, Goldman, and Scott~\cite{dooly1998tcp, dooly2001line} proved that the following natural (in hindsight) Greedy algorithm achieves a competitive ratio of 2: acknowledge whenever the delay cost has increased by the cost of an acknowledgment. This is also the optimal deterministic competitive ratio as TCP with additive delay generalizes the well-known Ski Rental problem for which a lower bound of 2 is known \cite{KarlinMRS88}. Subsequently, Seiden~\cite{Seiden00} established a lower
bound of $e/(e - 1)\approx 1.58$ on the competitive ratio of randomized algorithms; a
matching upper bound was shown by Karlin, Kenyon, and Randall~\cite{KarlinKR01}
shortly after. These algorithms also work for slightly more general delay costs: each request $j$ 
incurs a delay cost $c_j$ that is a non-decreasing function of its waiting time,
and the overall delay cost is $\sum_{j=1}^n c_j$. We call this more general problem \emph{TCP with additive delay}, since the delay cost is the sum of the delay costs of the requests. 

\paragraph{Beyond TCP.}
TCP with additive delay is the simplest problem in the class of \emph{online problems with delay}, which captures trade-offs between batching requests to reduce service costs and minimize delays. An online problem with delay involves a sequence of requests that arrive online and a \emph{service cost function} $g$ which is a function on the set of requests; the goal is to partition the requests into (possibly non-contiguous) batches and a service time $t_i$ for each batch $\sigma_i$ so as to minimize the objective function $\sum_i g(\sigma_i) + \bdelay(\sigma_i,t_i)$. The sum $\sum_i g(\sigma_i)$ is called the \emph{service cost}.

Most prior work on online problems with delay has focused on objective functions with more general service costs, with the overall delay cost computed as the sum of the delay costs of requests. These objective functions capture various types of economies of scale arising in varied practical domains such as scheduling, supply chain management, and logistics. The ones most relevant to our work are: joint replenishment~\cite{BuchbinderKLMS13,BienkowskiBCJS13, chen2022online,gyorgyi2023joint, moseley2025putting, ShmoysSU26, AzarL26, DBLP:conf/icalp/DinitzFU26}, multi-level aggregation~\cite{BritoKV04,BuchbinderFNT17, bienkowski2020online, BienkowskiBBCDF21, le2023power}, network design with delay~\cite{frameworkDelay, azar2020beyond,Touitou23}, set cover~\cite{AzarCKT20,Touitou21,DBLP:journals/algorithmica/EzraLPRU26}, set aggregation~\cite{CarrascoPSV18}, online service with delay~\cite{azar2017online,DBLP:conf/sirocco/BienkowskiKS18,DBLP:conf/isaac/KrneticM0W20,DBLP:conf/focs/00010P21,DBLP:journals/siamcomp/GuptaKP22, DBLP:conf/stoc/Touitou23,DBLP:conf/stacs/Touitou25}, and matching~\cite{emek2016online,AzarCK17, ashlagi2017min, bienkowski2017match,DBLP:journals/tcs/EmekSW19,azar2021min,DeryckereU23,DBLP:conf/wine/HeLSWWZ23,DBLP:journals/mst/MariPRS25, DBLP:conf/soda/DufayW26,GSU26}. 

\paragraph{Beyond Additive Delay Cost.}
In many of the above practical domains, it is natural to consider non-additive delay costs, where the penalty incurred by delaying requests is not simply the sum of their waiting times or delay costs. For example, one may wish to control tail latency by ensuring that no request waits too long, or, more generally, by penalizing the largest few waiting times. In settings with service-level agreements, another natural objective is to limit the number of requests that experience positive delay or exceed a prescribed delay threshold. Such objectives capture global quality-of-service or incident-based costs that are not well represented by additive delays.

The main goal of our work is to initiate a systematic study of online problems with delay beyond additive delay costs, starting with the setting that has the simplest service cost function: online TCP Acknowledgment. We investigate the following broad questions.
\begin{question}\label{ques1}
    For what families of natural delay costs can we obtain a good competitive ratio for the online TCP acknowledgment problem?
\end{question}

As Greedy is heavily used as a subroutine in many algorithms for online problems with delay to determine service times (e.g.~in \cite{azar2020beyond,chen2022online}), a good understanding of the power of Greedy can help us extend those algorithms to work with non-additive delay costs.

\begin{question}\label{ques2}
    For what families of natural delay costs is Greedy a good algorithm for the online TCP Acknowledgment problem?
\end{question}

To the best of our knowledge, the only previous works on online problems with delay that go beyond additive delay costs are those below. 

\paragraph{General Batch Delay Cost.}
Dooly, Goldman, and Scott~\cite{dooly1998tcp,dooly2001line} considered a far-reaching generalization of the objective in \cref{eq:obj} where the batch delay cost function $\bdelay(\cdot,\cdot)$ is only assumed to be monotone non-decreasing: if $\sigma$ is a subsequence of  $\sigma'$ and $t \leq t'$, then $\bdelay(\sigma,t) \leq \bdelay(\sigma', t)$ and $\bdelay(\sigma,t) \leq \bdelay(\sigma,t')$. Equivalently, the marginal delay cost at each timestep depends only on the set of currently pending requests. These general batch-delay cost functions can be used to model not only the goal of avoiding requests waiting too long but also the goal of avoiding too many pending requests at any time. The former can be modeled using $\ell_p$-norms, while the latter can be modeled by taking $\bdelay(\sigma_i,t_i)$ to be a function of $|\sigma_i|$.\footnote{This type of delay cost function was also studied under the name ``size-based delay'' by Deryckere and Umboh~\cite{DeryckereU23} and Gan, Sun, and Umboh~\cite{GSU26} for online matching with delay.} Dooly, Goldman, and Scott showed that Greedy is 2-competitive when $\bdelay(\sigma_i,t_i) = \max_{j \in \sigma_i}(t_{i} - a_i)$, i.e.~the batch delay cost is the maximum waiting time of the batch. They also claimed that Greedy is 2-competitive for monotone, nondecreasing batch-delay costs. However, we show that Greedy is actually $\Omega(n)$-competitive, and thus new algorithms are required.

\paragraph{Penalizing Long Delays.}
Later, Albers and Bals~\cite{AlbersB03, AlbersB05} explored a structurally different objective function of the form \[k + \max_{i=1}^k \max_{j \in \sigma_i}(t_i-a_j)^p\] for some integer $p \geq 1$. When $p=1$, the delay cost is simply the maximum waiting time of any packet. They gave tight upper and lower bounds on the competitive ratio of deterministic online algorithms: when $p=1$, it is $\pi^2/6 \approx 1.644$; when $p > 1$, it is an expression that approaches $1.5$ as $p$ tends to infinity. For randomized algorithms, they showed lower bounds of $3/(3 - 2/e) \approx 1.324$ and $2/(2 - 1/e) \approx 1.225$, when $p=1$ and when $p > 1$, respectively. A comprehensive overview of this area until 2003 can be found in Chapter 7.2 of the survey by Albers~\cite{albers2003online}.

\subsection{Our Contributions}

In this work, we make significant progress towards answering the above questions. We consider three broad classes of delay costs that generalize those considered in previous works: \emph{max-monotone}, \emph{sum-monotone}, and \emph{batch-oblivious}.

\paragraph{Max-Monotone Delay Cost.}
We begin by considering a generalization of the delay cost considered by Albers and Bals~\cite{AlbersB05}. In the \emph{max-monotone delay} setting, the objective function is
\begin{equation}
  \tag{Max-Monotone}
    k + \max_{i=1}^k \bdelay(\sigma_i,t_i)
  \end{equation}
for some monotone non-decreasing $\bdelay(\cdot,\cdot)$.  We show that Greedy is $2$-competitive without any further restrictions on $\bdelay(\cdot,\cdot)$ using a similar analysis as Greedy under additive delay costs.

\begin{restatable}{theorem}{thmmaxmonotone}
  \label{thm:max-monotone}
  Greedy is a deterministic $2$-competitive algorithm for TCP Acknowledgment with max-monotone delay cost.
\end{restatable}

\paragraph{Sum-Monotone Delay Cost.}
Our main technical contributions are for the \emph{sum-monotone delay} setting, where the objective is
\begin{equation}
  \tag{Sum-Monotone}
  k + \sum_{i=1}^k \bdelay(\sigma_i,t_i)
\end{equation}
for some monotone non-decreasing $\bdelay(\cdot,\cdot)$. 

Intuitively, given the success of Greedy for max-monotone delay, and the apparent simplicity of the problem---that at any point in time, we only need to decide between two actions (to acknowledge or not)---one expects that Greedy is also good for sum-monotone. Indeed, Dooly, Goldman, and Scott~\cite[Theorem 10]{dooly2001line} claimed that the same analysis of Greedy under additive delays can be extended to sum-monotone delays. Surprisingly, we show that Greedy is actually terrible.

\begin{theorem}[Informal; see \cref{thm:greedy-fail}]
  \label{thm:greedy-fail-inf}
  Greedy is $\Omega(n)$-competitive for TCP Acknowledgment with sum-monotone delay cost.
\end{theorem}

A further surprise is that there is no deterministic constant-competitive algorithm for sum-monotone delays. We give a lower bound of $\Omega(\log n)$ on the deterministic competitive ratio for sum-monotone delays via the Parking Permit problem---a generalization of Ski Rental---and a matching upper bound.

\begin{restatable}{theorem}{thmsummono}
\label{thm:summono}
  The deterministic competitive ratio for TCP Acknowledgment with sum-monotone delay cost is $\Theta(\log n)$.
\end{restatable}

Note that we measure the waiting time of a packet $j \in \sigma_i$ as $t_i - a_j$. However, all of our results also apply if the waiting time is transformed by any monotone, nondecreasing (possibly nonlinear) function before aggregation. Since all delay cost functions we consider are monotone non-decreasing, such transformations can equivalently be absorbed into the definition of the batch delay cost function $\bdelay(\cdot,\cdot)$.

\paragraph{Batch-Oblivious Delay Cost.}
In the \emph{batch-oblivious delay} setting, the delay cost is given by applying a monotone non-decreasing function $f$ to the \emph{delay vector} $d$ of waiting times of each packet. More precisely, the $j$-th coordinate of the vector $d$ is the waiting time of the $j$-th packet. To summarize, the batch-oblivious objective is
\begin{equation}
  \tag{Batch-Oblivious}
  k + f(d).
\end{equation}
Unlike max-monotone and sum-monotone delays, batch-oblivious delays are ``global'' objectives. For example, we can measure the $\ell$ largest waiting times among all requests by taking $f$ to be the $\textsf{Top-}\ell$ norm;\footnote{The $\textsf{Top-}\ell$ norm of a vector is the sum of its $\ell$-largest entries.} note that this objective is not captured by sum-monotone and max-monotone delays.

In this setting, we show that a natural extension of Greedy is $2$-competitive not only for $\ell_p$ norms and $\textsf{Top-}\ell$ norms, but also for a broader class of functions that are \emph{continuous submodular}. The proof is much subtler than for additive delay.

\begin{restatable}{theorem}{submodular}
\label{thm:greedy}
  Greedy is a deterministic $2$-competitive algorithm for TCP Acknowledgment
  with batch-oblivious continuous submodular delay cost.
\end{restatable}

Continuous submodularity (\cref{def:cs}) generalizes the notion of submodularity for discrete functions to functions on real vectors~\cite{francis}. Patton, Russo, and Singla~\cite{PattonR023} observed that $\ell_p$ norms and several other symmetric norms are continuous submodular. They also showed that any symmetric norm can be approximated by a continuous submodular norm up to a logarithmic factor. Using the results of \cite{PattonR023}, we get the following result.

\begin{restatable}{corollary}{normthm} %
  There is a deterministic 2-competitive algorithm for TCP Acknowledgment, where the delay cost is: an $\ell_p$ norm, a $\textsf{Top-}\ell$ norm, or an ordered norm. When the delay cost is a symmetric norm $\norm$, then there is a deterministic $O(\log \rho)$-competitive algorithm where $\rho = \norm[(1,1,\ldots,1)]/\norm[(1,0,\ldots,0)] \leq n$. In particular, there is a deterministic $O(\log n)$-competitive algorithm when the delay cost is a monotone symmetric norm.
\end{restatable}

Our final result shows that if we venture too far from norms to concave functions, then the problem becomes significantly more difficult.\footnote{Recall that norms are convex.} 
\begin{restatable}{theorem}{concavelb}
  \label{thm:concave-lb}
  Every deterministic algorithm is $\Omega(\sqrt{n})$-competitive for TCP Acknowledgment with batch-oblivious concave delay cost.
\end{restatable}

\subsection{Our Techniques}
\label{sec:techniques}

We first develop intuition for our algorithms by recalling the Greedy algorithm and its analysis for the classic additive delay costs. Then, we describe the modifications required to prove that Greedy is $2$-competitive for max-monotone and continuous submodular delay costs. Finally, we give an overview of our algorithmic framework for sum-monotone delays, where greedy strategies are no longer competitive, and outline the phase-based approach that overcomes this barrier, together with the accompanying logarithmic lower bound.

\paragraph{Recap: Additive Delay.}
The Greedy algorithm can be easily described as follows: send an acknowledgment when the total increase in its delay cost since the last acknowledgment equals $1$---the cost of an acknowledgment. For additive delays, this is equivalent to sending an acknowledgment when the total waiting time of unacknowledged requests since the last acknowledgment reaches $1$. Let us now analyze the competitive ratio of Greedy for additive delays. Suppose Greedy makes $k$ acknowledgments at times $t_1 < \ldots < t_k$, with batches $\sigma_1, \ldots, \sigma_k$. The analysis consists of two main components:
\begin{enumerate}
  \item By design, Greedy's acknowledgment cost equals its delay cost, so its cost is exactly $2k$.
  \item The optimal solution pays at least $1$ during each of the $k$ intervals $I_j = [t_{j-1}, t_j)$ for $j \in [k]$, where $t_0=0$. This is because with additive delays, if the optimal solution does not send an acknowledgment within the time interval $I_j$, then the packets in batch $\sigma_i$ contribute $\bdelay(\sigma_i,t)$ to its delay cost.
\end{enumerate}
While the first component still holds for any monotone delay cost, the second component breaks for the delay costs that we consider.

\paragraph{Max-Monotone Delay.}
For max-monotone delays, we get that $\bdelay(\sigma_j,t_j) = j$ and if the optimal solution does not send an acknowledgment during interval $I_j$, then the packets of $\sigma_j$ were acknowledged together at some time $t^{*}_j \geq t_j$, and so by monotonicity of $\bdelay(\cdot,\cdot)$, the delay cost of the optimal solution is at least $j$. Thus, if the optimal solution makes $k^{*} < k$ acknowledgments, its delay cost is at least $k - k^{*}$, and so its total cost is at least $k$.

\paragraph{Sum-Monotone Delay.}
For sum-monotone delay costs, greedy strategies are not competitive. We show that for any constant threshold, there exists an instance where a greedy algorithm incurs a cost that is $\Omega(n)$ times larger than that of the offline optimum. This fundamental limitation stems from the additive nature of the delay cost: greedy algorithms react too locally, failing to recognize when it would be better to wait and serve a larger group of packets with a single acknowledgment. For instance, if the delay function is concave and grows steadily up to the acknowledgment cost, and if the packets are spaced far enough apart, greedy will serve each packet individually --- incurring an acknowledgment cost each time --- while the optimal solution serves all at once at the last packet arrival, paying the delay cost only once and saving on acknowledgments.

To overcome this, we use our offline dynamic programming solution to guide online decisions. The key idea is to monitor the cost of the optimal solution on the packets seen so far. We partition the execution of the algorithm into consecutive \emph{services}, each of which aggregates packets and ends with an acknowledgment. At the start of a service, we examine the sequence of packets observed so far and identify the longest suffix for which the offline optimum issues a single acknowledgment. The cost of optimally serving this suffix determines a \emph{budget} for the service. The algorithm then continues to accumulate packets and issues an acknowledgment as soon as the total delay cost exceeds this budget. Thus, the service budget is derived directly from the structure of an optimal offline solution on the currently observed packets.

To avoid repeatedly issuing acknowledgments with small accumulated delay costs, each budget service is followed by three \emph{buffer services} whose budgets are larger by a constant factor than that of the preceding budget service. These buffer services separate consecutive budget services and later allow us to analyze their costs independently. Throughout the execution, the algorithm continues to monitor newly observed suffixes. During a budget service, the arrival of any longer critical suffix triggers an update of the service budget, whereas during a buffer service, such an update occurs only if the offline cost of a suffix increases by a constant factor. This controlled budget growth ensures that the number of services grows only logarithmically, implying that the total cost incurred by the algorithm is within an $O(\log n)$ factor of the offline optimum.

\paragraph{Lower Bound for Sum-Monotone Delay.}
Observe that if the batch delay cost depends only on the maximum waiting time of the requests in the batch, then the offline problem essentially reduces to an interval covering problem, where we seek to cover the arrival times with intervals whose cost is a function of the interval length. Thus, it makes sense to consider the Parking Permit Problem~\cite{Meyerson05}, a generalization of the classic Ski Rental problem. In this problem, we are given different permit types and for each permit type $k$, its cost $C_k$ and duration $D_k$. Purchasing permit $k$ at time $t$ covers the time up to $t + D_k$. Online, a sequence of requests over time, and for each requested time $t$, we need to ensure that it is covered by a purchased permit. A deterministic lower bound of $\Omega(\log n)$ is known for this problem \cite{NaorUW22}.

Offline, it is straightforward to reduce Parking Permit to our problem, since both problems are interval covering problems in which the cost of an interval depends only on its length. However, it is unclear how to do the reduction online. Intuitively, we would like to use an algorithm for our problem to solve Parking Permit as follows: when a request arrives for Parking Permit at time $t$, we give the same request to the TCP Acknowledgment algorithm, see how long it waits before acknowledging the request, and buy a permit of the corresponding length. Unfortunately, at the time the request arrives, the TCP Acknowledgment algorithm does not need to commit to when it will acknowledge the request. Indeed, if the next request arrives immediately after, it may be acknowledged at a different time than if it arrives much later. We leave the possibility of an online reduction from Parking Permit to our problem as an open, tantalizing problem.

Fortunately, the deterministic adversarial sequence for Parking Permit ensures that the next requested time is later than all the permits the algorithm has purchased so far. Thus, when the Parking Permit request arrives at time $t$, we pass it to the TCP Acknowledgment algorithm and simulate it into the future to determine at what time $t^*$ it would acknowledge the request, assuming no further requests, and purchase the corresponding permit that covers up to time $t^*$. Since the next Parking Permit request arrives after the permit expires, the next TCP Acknowledgment request also arrives after $t^*$. Since the TCP Acknowledgment algorithm is online and the actual request sequence up to time $t^*$ matches the simulation, it also acknowledges at time $t^*$. Thus, the TCP Acknowledgment algorithm and the Parking Permit algorithm behave identically, yielding the desired lower bound.

\paragraph{Continuous Submodular Delay.}
The proof of \cref{thm:greedy} requires a more subtle analysis of Greedy. The second part of the greedy analysis for additive delays breaks down even when the continuous submodular delay cost is a concave function of the total waiting time. Intuitively, this is because the optimal solution can ``front-load'' its delay cost by acknowledging, e.g., the first few packets a long time after their arrival. This increases the delay cost, so later packets contribute only a small marginal amount to it. To address this, we consider a hypothetical solution that produces the same acknowledgments as the optimal solution, but each packet's waiting time is capped at its waiting time in Greedy. Since $f$ is monotone, capping the waiting times can only decrease the delay cost. Thus, the cost of the hypothetical solution is at most the cost of the optimal solution, so it suffices to prove a lower bound on the cost of the hypothetical solution. The capping prevents the hypothetical solution from front-loading its delay cost and enables us to argue that if it does not make an acknowledgment during interval $I_j$, then its delay cost increases by at least $1$ during the interval.

\section{Preliminaries}

\paragraph{Problem Setup.}
We consider the \emph{Online TCP Acknowledgment} problem. The input is a sequence of packet arrival times $A = ( a_1, a_2, \dots, a_n )$, where $a_j$ is the time at which the $j$-th packet arrives, and we assume $a_1 \le a_2 \le \cdots \le a_n$. An \emph{online algorithm} must select acknowledgment times $t_1 < t_2 < \cdots < t_k$, such that every packet is acknowledged at or after arrival --- that is, for each $1 \le j \le n$, there exists some $1 \le i \le k$ with $a_j \le t_i$. Each acknowledgment $t_i$ covers the batch $\sigma_i = \{\, j : t_{i-1} < a_j \le t_i \,\}$, where we set $t_0 = 0$. These batches form a partition of the packet indices $\{1,\dots,n\}$. Each packet $j \in \sigma_i$ experiences a \emph{waiting time} (or \emph{delay}) $d_j = t_i - a_j$, and we denote the full delay vector by $d = (d_1, \dots, d_n)$. The goal is to minimize the total cost, which typically balances two competing objectives: reducing the number of acknowledgments and limiting the total incurred delay. The online nature of the problem is that, as time $t$ increases, the algorithm is aware only of requests whose arrival time is at most $t$ and can only make an acknowledgment at the current time $t$.

\paragraph{Batch-Aware TCP.}
The \emph{Batch-Aware TCP} model defines the delay cost in terms of the individual acknowledgment batches. Let $\bdelay(\sigma_i, t_i)$ denote the delay cost of acknowledging batch $\sigma_i$ at time $t_i$. The total delay cost is then computed from the per-batch costs, thereby capturing richer dependencies between the acknowledgment structure and delay.

In this work, we focus on two important classes of batch-aware cost models. In the \emph{sum-monotone delay} model, the total cost is given by $k + \sum_{i=1}^k \bdelay(\sigma_i, t_i)$. Here, the function $\bdelay(\cdot, \cdot)$ is required to be monotone non-decreasing: if one batch is a subsequence of another and the acknowledgment times are ordered accordingly, the delay cost cannot decrease. This model captures objectives in which the overall delay accumulates across batches and scales with the batch size or delay duration.

In the \emph{max-monotone delay} model, the objective is $k + \max_{1 \le i \le k} \bdelay(\sigma_i, t_i),$ where $\bdelay(\cdot, \cdot)$ is any monotone non-decreasing function. Unlike the sum-based case, this model focuses on minimizing the worst per-batch delay contribution, which is natural in applications where fairness or tail performance is critical.

\paragraph{Batch-Oblivious TCP.}
In the \emph{Batch-Oblivious TCP} model, the delay cost is defined globally over the entire delay vector $d$. Specifically, the delay cost is given by a function $f_n(d_1, \dots, d_n)$, where $f_n : \mathbb{R}_+^n \to \mathbb{R}_+$ is assumed to be monotone and consistent under zero-padding. Since the number of packets $n$ is not known in advance, we assume the algorithm learns the function $f_i$ only when the $i$th packet arrives. These functions satisfy the consistency condition
\[
    f_{i+1}(d_1, \dots, d_i, 0) = f_i(d_1, \dots, d_i),
\]
ensuring that appending a packet with zero delay does not affect the cost. A natural example of this model is total waiting time, where $f_n(d) = \sum_{j=1}^n d_j$. The overall objective in this setting is to minimize $k + f_n(d)$, where $k$ denotes the number of acknowledgments. We refer to this setting as batch-oblivious because the delay cost treats all packet delays globally, regardless of which acknowledgment batch a packet belongs to.

\section{Max-Monotone Delay Cost}
\label{sec:max-monotone}

We begin our investigation with the max-monotone delay model, where the total cost is given by the number of acknowledgments plus the maximum per-batch delay cost. We show that the natural greedy algorithm is 2-competitive.

\thmmaxmonotone*

\begin{proof}
Recall that we assume the acknowledgment cost is normalized to 1. We analyze the natural greedy algorithm that sends an acknowledgment whenever the total cost increases by 1. Since we are in the max-monotone model, this is equivalent to acknowledging a batch of packets once its delay cost reaches the current batch index. For example, the algorithm sends the first acknowledgment when the currently pending packets accumulate delay cost 1, then the second when the pending batch reaches delay cost 2, and so on. In this way, each acknowledgment increases the total cost by exactly 1 from the delay, in addition to the unit cost of the acknowledgment.

Let $t_j$ denote the time when the algorithm makes its $j$-th acknowledgment, and define the interval $T_j = (t_{j-1}, t_j]$, with $t_0 = 0$. Moreover, let $\sigma_j$ be the batch of packets acknowledged at time $t_j$. By design, we have $\bdelay(\sigma_j, t_j) = j$, since the algorithm waited until this batch’s delay cost reached exactly $j$. Thus, the total cost of the algorithm is $k$ from acknowledgments and $\max_{1 \le j \le k} \bdelay(\sigma_j, t_j) = k$ from delay, totaling at most $2k$.

Denote by $k^*$ the number of acknowledgments made by the optimal solution. We now argue that its total cost is at least $k$, which implies that the greedy algorithm is 2-competitive.

\begin{claim}
\label{clm:max-mono}
If the optimal solution does not send an acknowledgment during interval $T_j$, then its delay cost is at least $j$.
\end{claim}

\begin{proof}
Since no acknowledgment is made during $T_j$, all packets in $\sigma_j$ must be acknowledged by the optimal solution in a later batch $\sigma^*$, sent at time $t^* \ge t_j$. Because $\bdelay(\cdot,\cdot)$ is monotone non-decreasing in both batch size and acknowledgment time, it follows that $\bdelay(\sigma^*, t^*) \ge \bdelay(\sigma_j, t_j)$, which equals $j$, proving the claim.
\end{proof}

Suppose that the optimal solution makes $k^{*}$ acknowledgments. If $k^{*} \geq k$, then we are done. Otherwise, \cref{clm:max-mono} implies that the delay cost of the optimal solution is at least $k - k^{*}$ (as the acknowledgments can be placed in the last $k^*$ $T_i$ intervals) and so the optimal solution has total cost at least $k$, as desired.
\end{proof}

\section{Sum-Monotone Delay Cost}

In this section, we focus on TCP acknowledgment under the sum-monotone delay-cost model. First, we prove that any deterministic greedy policy is $\Omega(n)$-competitive: there exists an input of size $n$ that enforces the greedy algorithm to pay $n$-fold more than the offline optimum. We then turn to the offline problem and present a polynomial-time algorithm via dynamic programming, that computes an exact optimum schedule. By embedding this offline oracle into an online framework, we obtain a new algorithm with an $O(\log n)$ competitive ratio: at each decision point, it computes the offline optimum schedule for the set of all packets that have arrived so far and uses that schedule to guide its acknowledgments. Finally, we show that every deterministic algorithm is $\Omega(\log n)$-competitive.

\subsection{Lower Bound for Greedy}
We begin with a more precise statement of \cref{thm:greedy-fail-inf}.
Elsewhere in this paper, when we talk about the Greedy algorithm, we refer to
the algorithm that sends an acknowledgment whenever $\bdelay(\sigma,t) = 1$,
where $\sigma$ is the set of currently unacknowledged packets and $t$ is the
current time. Our lower bound holds for a family of greedy algorithms: for $\tau > 0$, algorithm
$\greedytau$ sends an acknowledgment when $\bdelay(\sigma,t) = \tau$. 

\begin{theorem}
  \label{thm:greedy-fail}
  For every $\tau > 0$, the algorithm $\greedytau$ is $\Omega(n)$-competitive for TCP Acknowledgment with
  sum-monotone delay cost.  
\end{theorem}

\begin{proof}
 Fix $\tau > 0$. Define the batch delay function $\bdelay(\sigma,t) =
 \min\{\sum_{i \in \sigma} d_i, \tau\}$, where $d_i = a_i - t$. Consider the sequence of $n$ packet with arrival times $a_i = i(\tau + \epsilon)$ for a small enough $\epsilon$.
 Observe that $\greedytau$ makes an acknowledgment exactly $\tau$ time steps
 after the arrival of each packet. Thus, $\greedytau$ makes $n$ acknowledgments
 and incurs a delay cost of $n\tau$ for a total cost $n(1+\tau)$. 
  On the other hand, the solution that waits to make a single acknowledgment at
  the arrival time of the last packet $a_n$ has a total cost of at most $1 +
  \tau$. Thus, the competitive ratio of $\greedytau$ is at
  least $n$.
\end{proof}

\subsection{Offline Model}

Solving the offline version of the TCP acknowledgment under the sum-monotone delay-cost model can be efficiently performed using dynamic programming. The general structure of the solution closely follows the approach used by Dooly et al.~\cite{dooly2001line}, though here we consider a general batch cost function $\bdelay(\sigma, t)$, making our model a natural generalization of their framework.

Let $\sigma$ be the given sequence of packets, and $\sigma|_i^j$ be its subsequence containing packets $i$ through $j$. We define the subproblems for our dynamic programming solution as follows. Let $\DP[i]$ denote the minimum total cost of acknowledging packets $1$ through $i$, and set $\DP[0] = 0$. Then, the optimal cost can be computed using the following recurrence relation:
\[
    \DP[i] = \min_{1 \leq j \leq i} \left\{ \DP[j-1] + \bdelay\left(\sigma|_j^i, a_i\right) + 1 \right\}.
\]
Intuitively, at each step $i$, the algorithm evaluates all possible previous acknowledgment positions $j$ and selects the one that minimizes the sum of the optimal cost up to $j-1$, the delay cost $\bdelay(\sigma|_j^i, a_i)$, and the acknowledgment cost $1$.

The correctness of this approach follows directly from the optimal substructure property: an optimal partitioning of packets into acknowledged subsequences necessarily implies that each of these subsequences itself must be an optimal solution to the corresponding subproblems. The computational complexity of this method is $O(n^2)$, since evaluating each $\DP[i]$ involves examining $O(n)$ previous acknowledgment points, and we have $n$ such subproblems.

\subsection{Online Model}

Using the offline result from the previous subsection, we now prove the following lemma.

\begin{restatable}{lemma}{thmnested}
    \label{thm:nested}
    There is a deterministic $O(\log n)$-competitive algorithm for TCP
    Acknowledgment with sum-monotone delay cost.
\end{restatable}

The key idea is to use the cost of the optimal offline solution on the sequence of packets seen so far to guide when our algorithm issues acknowledgments. Suppose the $j$-th packet arrives and there are no pending packets; we then start a new \emph{budget service} $s$. Let $\sigma|^j$ denote the sequence of the first $j$ packets observed so far.\footnote{Here we slightly overload notation: earlier, $\sigma|^j$ denoted the $j$-th batch served by the algorithm.} We first identify the longest \emph{critical suffix} of $\sigma|^j$, that is, a suffix $\sigma_s \subseteq \sigma|^j$ such that the optimal offline strategy for $\sigma_s$ consists of issuing a single acknowledgment exactly at the arrival time of its last packet. Observe that at least one such suffix always exists, namely the one containing only packet $j$. Moreover, notice that for a critical suffix $\sigma_s$, the optimal cost of serving it equals $\bdelay(\sigma_s, a_j) + 1$. To simplify notation, we define $\bserve(\sigma_s, a_j) = \bdelay(\sigma_s, a_j) + 1$. Whenever an acknowledgment is issued at the arrival time of the last packet in the suffix, we omit the second argument and write $\bserve(\sigma_s)$.

Based on $\sigma_s$, we define a budget $b_s$ for service $s$ (its precise definition will be given later). The service $s$ then aggregates arriving packets and issues an acknowledgment as soon as their total accumulated delay cost reaches $b_s$; we refer to this moment as the \emph{critical time} of service $s$. If, during the aggregation phase of a budget service, a new packet arrives, we again compute the longest critical suffix $\sigma' \subseteq \sigma|^j$. If $\sigma'$ strictly extends $\sigma_s$ (that is, if it starts earlier than $\sigma_s$), we update the budget $b_s$ accordingly to reflect the increased cost suggested by the offline optimum.

When a budget service $s$ finally issues its acknowledgment, the algorithm schedules three consecutive \emph{buffer services}, each with a budget equal to twice the budget of $s$. Each buffer service aggregates packets and acknowledges them as soon as their delay cost reaches its assigned budget, and this process is repeated three times.\footnote{The purpose of these buffer services is to separate consecutive budget services and simplify the analysis.} During a buffer service, if a packet arrives that induces a critical suffix whose cost is at least twice the cost of $\sigma_s$, then the current service is immediately promoted to a new budget service, with its budget defined accordingly.

Assigning exponentially larger budgets to buffer services prevents the algorithm from issuing many low-cost acknowledgments in situations where the offline optimum would issue a single high-cost acknowledgment. By scaling service budgets exponentially in response to increasing costs of critical suffixes, the algorithm ensures that its total cost increases only logarithmically relative to the optimal offline cost, yielding an $O(\log n)$ competitive ratio.

In what follows, we use $\OPT(\sigma_s)$ to denote the cost of the optimal offline solution for batch $\sigma_s$.

\subsubsection{Algorithm Description}

We begin by formalizing the notion of \emph{critical batches}, which correspond to the critical suffixes discussed earlier and will play a central role in the design and analysis of our algorithm.

\begin{definition}[Critical Batch] \label{def:critical_set}
    We call batch $\sigma$ \emph{critical} if it satisfies $\OPT(\sigma) = \bserve(\sigma)$.
\end{definition}

We are now ready to present \cref{alg:tcp_sum_monotone_delay}. This algorithm is structured around two event-driven functions: `UponArrival`, which is triggered whenever a new packet arrives, and `UponCriticalTime`, which is triggered when a precomputed critical time is reached. Together, these functions manage the scheduling of acknowledgments based on the observed input structure.

\begin{algorithm}[ht]
\caption{\label{alg:tcp_sum_monotone_delay} Sum-Monotone TCP Algorithm}

\begin{description}[labelwidth=2.5cm, labelsep=0.5cm, leftmargin=3cm]
    \item[Initialization:] sequence of already seen packets $\sigma = []$, critical time $t_c = 0$, \\ critical batch $\sigma_c = []$ and budget $b$
\end{description}

\begin{algorithmic}[1]

\functionHeader{UponArrival($r$)}{}

\STATE put $r$ at the end of $\sigma$

\STATE let $\sigma'$ be the largest critical suffix

\STATE \textbf{if} there is no service running, i.e., $\sigma_c$ is empty \textbf{then}   \hfill $\triangleright$ start a buffer service

\STATE \ \ \ $\sigma_c = \sigma'$   \label{alg_line:new_batch_no_service}

\STATE \textbf{else if} there is a budget service running \textbf{then}

\STATE \ \ \ \textbf{if } $\sigma_c \subseteq \sigma'$ \textbf{then}   \hfill $\triangleright$ check if the current critical suffix is a larger   \label{alg_line:cond_increase_budget}

\STATE \ \ \ \ \ \ $\sigma_c = \sigma'$   \hfill $\triangleright$ if so, update the critical batch   \label{alg_line:new_batch_increase_budget}

\STATE \textbf{else if} there is a buffer service running \textbf{then}

\STATE \ \ \ \textbf{if } $\bserve(\sigma') \geq 2 \cdot \bserve(\sigma_c)$ \textbf{then}   \hfill $\triangleright$ check whether the critical suffix is large enough   \label{alg_line:cond_buffer_to_budget}

\STATE \ \ \ \ \ \ $\sigma_c = \sigma'$   \hfill $\triangleright$ if so, update the critical batch

\STATE \ \ \ \ \ \ switch the current service to a budget service
\label{alg_line:switch_service_type}

\STATE $b = 2 \cdot \bserve(\sigma_c)$   \hfill $\triangleright$ update the critical batch   \label{alg_line:set_the_budget}

\STATE let $t_c$ be such that $\bserve(\text{pending packets}, t_c) = b$   \hfill $\triangleright$ update the critical time
\label{alg_line:update_critical_time}

\vspace{14pt}

\functionHeader{UponCriticalTime()}{$\triangleright$ called when $t_c$ is reached}

\STATE make an acknowledgment

\STATE \textbf{if} the budget service ended \textbf{then}

\STATE \ \ \ start a new buffer service and double the budget $b \leftarrow 2b$   \label{alg_line:double_budget_for_buffer}

\STATE \textbf{else if} the first or second buffer service ended \textbf{then}

\STATE \ \ \ start a consecutive buffer service

\STATE \textbf{else if} the third buffer service ended \textbf{then}

\STATE \ \ \ $\sigma_c = []$, $b = 0$   \hfill $\triangleright$ reset the variables

\end{algorithmic}
\end{algorithm}

When a packet arrives, the \texttt{UponArrival} appends it to the already-seen packets sequence and finds the largest critical suffix. Then, depending on whether there is a service running and if so --- of what type it is --- we possibly update the budget of the service (and its type, as in line \ref{alg_line:switch_service_type}). Notice that line \ref{alg_line:update_critical_time} sets the critical time $t_c$ so that acknowledging all currently pending packets at time $t_c$ will incur a delay cost exactly equal to the budget $b$. If this cost never reaches $b$, we assume $t_c = \infty$, meaning the current service will conclude only at the end of the input sequence.

The \texttt{UponCriticalTime} function is called when $t_c$ is reached. It performs an acknowledgment and issues a new service, or resets the state if all buffer services have completed.

\subsubsection{Critical Batches and Relations Between Them}

We begin by introducing some notation. Suppose that \cref{alg:tcp_sum_monotone_delay} issues $k$ acknowledgments for budget services over a sequence of $n$ requests. We denote the corresponding critical batches by $\sigma_c^i$ for $i \in \{1,2,\ldots,k\}$. Each critical batch may be followed by up to three buffer services. We denote by $\sigma_{b,l}^i$, for $l \in \{1,2,3\}$, the sequences of packets acknowledged by these services. If a particular buffer service is not started, we define the corresponding batch to be the empty sequence.

\begin{definition}[Critical Batch Intersection] \label{def:critical_set_intersection}
    We say that the critical batch $\sigma_c^j$ \emph{intersects} $\sigma_c^i$, $i < j$, if 
    \begin{enumerate}[label=(\arabic*)]
        \item $j = i+1$ and the budget service corresponding to $\sigma_c^j$ was started in line \ref{alg_line:switch_service_type} of the \texttt{UponArrival} function in \cref{alg:tcp_sum_monotone_delay}, or \label{cond:buffer_converted_service}
        \item $\sigma_{b,2}^i$ and $\sigma_{b,3}^i$ are both not empty, while $\sigma_c^j$ and $\sigma_{b,2}^i$ share at least one packet, that is, $\sigma_c^j \cap \sigma_{b,2}^i \not= \emptyset$. \label{cond:intersection_between_batches}
    \end{enumerate}
    When using set-theoretic notation, we implicitly identify a sequence with the set of packets it contains.
\end{definition}

Understanding how critical batches intersect is key to analyzing the proposed algorithm. To make our reasoning simpler, we introduce a few new terms to describe this relation.

\begin{definition}[Parent Batch] \label{def:parent_set}
    Given a critical batch $\sigma_c^i$, we define the set $\children(\sigma_c^i)$ to consist of all earlier critical batches that $\sigma_c^i$ intersects and for which $\sigma_c^i$ is the first subsequent critical batch to do so. For each batch $\sigma_c^j \in \children(\sigma_c^i)$, we refer to $\sigma_c^i$ as its \emph{parent batch}. We use the terms \emph{ancestor} and \emph{descendant} to denote the transitive closure of this relation.
\end{definition}

\begin{definition}[Final Critical Batch] \label{def:final_critical_set}
    We call a batch with no parent a \emph{final critical batch}.
\end{definition}

In what follows, we also use the following notation for the pivotal packet arrival times. We use $t_s^i$ and $t_c^i$ to denote the first and the last packet's arrival times in $\sigma_c^i$. Moreover, if $\sigma_c^i$ is a parent batch, we let $\sigma_e^i$ be its descendant that was acknowledged the earliest, and let $t_e^i$ be the arrival of the first packet in $\sigma_e^i$.

\begin{definition}[Coverage Interval]
    Given a critical batch $\sigma_c^i$, we define its \emph{coverage interval} as follows. If $\sigma_c^i$ has no children, then its coverage interval is $[t_s^i, t_b^i]$; otherwise, it is $[t_e^i, t_b^i]$. We denote this interval by~$I_c^i$. When convenient, we also use $I_c^i$ to denote the set of packets arriving within this time interval.
\end{definition}

\subsubsection{Analysis}

Let us start by showing how final critical batches relate to the optimal cost of serving packets.

\begin{lemma}
    Let $\sigma_c^1$, $\ldots$, $\sigma_c^\ell$ be a sequence of final critical batches obtained during a run of \cref{alg:tcp_sum_monotone_delay} on a packet sequence $\hat\sigma$. We assume that they are listed in the order in which the corresponding budget services were issued, and that the list contains all such batches. Then, it holds that
    \[
        \OPT\left(\hat\sigma\right) \geq \sum_{i=1}^\ell \OPT\left(I_c^i\right).
    \]
\end{lemma}
\begin{proof}
    Let $\hat\sigma_1, \ldots, \hat\sigma_p$ denote the partition of $\hat\sigma$ induced by the optimal solution, where each $\hat\sigma_i$ corresponds to a single acknowledgment. By definition, each $\hat\sigma_i$ forms a critical batch and therefore must be recognized by \cref{alg:tcp_sum_monotone_delay}.
    
    Suppose that for some $j \in \{1,\ldots,p\}$, the batch $\hat\sigma_j$ overlaps with at least two coverage intervals corresponding to the final critical batches identified by the algorithm. Let $\sigma_c^k$, for some $k \in \{1,\ldots,\ell\}$, be the latest such critical batch, and let $\hat{j}$ denote the last packet in $\hat\sigma_j$. Finally, let $\sigma_c^{k,j}$ be the largest critical batch among $\sigma_c^k$ and its descendants such that $\hat{j}$ belongs to one of the services associated with this batch.
    
    If $\hat{j}$ is served by the budget service associated with $\sigma_c^{k,j}$, then at the moment $\hat{j}$ arrives, the \texttt{UponArrival} function would recognize $\hat\sigma_j$ in either line~\ref{alg_line:new_batch_no_service} or~\ref{alg_line:new_batch_increase_budget} and act upon it. Otherwise, if $\hat{j}$ is served by a buffer service, it would trigger the start of a new budget service in line~\ref{alg_line:switch_service_type}. Indeed, since $\hat\sigma_j$ overlaps with all packets acknowledged by the budget service corresponding to $\sigma_c^{k,j}$, its service cost exceeds $2 \cdot \bserve(\sigma_c^{k,j})$, as stated in line \ref{alg_line:set_the_budget}, and thus the condition in line~\ref{alg_line:cond_buffer_to_budget} is satisfied.
    
    Therefore, no batch served by the optimal solution can overlap with more than one coverage interval. Consequently, for each coverage interval $I_c^i$, we consider the set of all batches $\hat\sigma_j$ that overlap with it. This set covers at least as many packets as $I_c^i$, and hence its total cost is at least $\OPT(I_c^i)$. Since each batch $\hat\sigma_j$ is counted at most once, the desired inequality follows.
\end{proof}

From there, we obtain two straightforward properties.

\begin{corollary} \label{cor:total_cost_vs_critical_sets}
    The total cost of \OPT{} on a given sequence of packets $\sigma$ is at least equal to the sum of the optimal costs of serving coverage intervals of all the final critical batches recognized during the run of \cref{alg:tcp_sum_monotone_delay} on $\sigma$.
\end{corollary}

\begin{corollary} \label{cor:parent_cost_equal_children}
    For any critical batch with children, the optimal cost of serving its coverage interval is at least the sum of the children's critical interval optimal serving costs.
\end{corollary}

Now, observe that also the following holds for any critical batch with children.

\begin{lemma} \label{lem:parent_child_doubling_cost}
    Let $\sigma_c^i$ be a critical batch such that $\children(\sigma_c^i) \not= \emptyset$. Then, for any child batch $\sigma_c^j$, we have $\bserve(\sigma_c^i) \geq 2 \cdot \bserve(\sigma_c^j)$.
\end{lemma}
\begin{proof}
    Since $\sigma_c^i$ is a parent of $\sigma_c^j$, \cref{def:parent_set} implies that $\sigma_c^i$ intersects $\sigma_c^j$. Therefore, one of the two conditions in \cref{def:critical_set_intersection}—namely, \ref{cond:buffer_converted_service} or \ref{cond:intersection_between_batches}—must hold.
    
    We first consider condition~\ref{cond:buffer_converted_service}. In this case, we have $j=i+1$, and the budget service corresponding to $\sigma_c^j$ is initiated in line~\ref{alg_line:switch_service_type} of the \texttt{UponArrival} function by some critical suffix $\sigma'$. Consequently, the condition in line~\ref{alg_line:cond_buffer_to_budget} holds, which implies that $\bserve(\sigma') \geq 2 \cdot \bserve(\sigma_c^i)$. Since $\sigma' \subseteq \sigma_c^j$ (as $\sigma'$ may have been extended in line~\ref{alg_line:cond_increase_budget} to become $\sigma_c^i$), the claimed inequality follows.
    
    We now turn to condition~\ref{cond:intersection_between_batches}. In this case, both $\sigma_{b,2}^i$ and $\sigma_{b,3}^i$ are non-empty, and $\sigma_c^j$ shares at least one packet with $\sigma_{b,2}^i$, that is, $\sigma_c^j \cap \sigma_{b,2}^i \neq \emptyset$. Since $\sigma_c^j$ overlaps with $\sigma_{b,2}^i$, it must also cover $\sigma_{b,3}^i$, which immediately follows $\sigma_{b,2}^i$. As $\sigma_{b,3}^i \not= \emptyset$, its service cost satisfies $\bserve(\sigma_{b,3}^i) = 2 \cdot 2 \cdot \bserve(\sigma_c^i)$, by lines~\ref{alg_line:double_budget_for_buffer} of the \texttt{UponCriticalTime} function and~\ref{alg_line:set_the_budget} of the \texttt{UponArrival} function. Hence, $\bserve(\sigma_c^j) \geq \bserve(\sigma_{b,3}^i) = 4 \cdot \bserve(\sigma_c^i)$, and the desired inequality also holds in this case.
\end{proof}

\begin{lemma} \label{lem:critical_set_service_cost}
    The budget service associated with the critical batch $\sigma_c^i$ incurs a cost of $2 \cdot \bserve(\sigma_c^i)$, while each buffer service associated with $\sigma_c^i$ incurs a cost of at most $4 \cdot \bserve(\sigma_c^i)$.
\end{lemma}
\begin{proof}
    By the definition of the algorithm.
\end{proof}

\begin{lemma} \label{lem:descendant_levels}
    Every final critical batch $\sigma_c^i$ has at most $\log(\OPT(I_c^i)) \leq \log n$ levels of descendant sets.
\end{lemma}
\begin{proof}
    By \cref{lem:parent_child_doubling_cost}, each child batch has a cost at least twice that of the parent. Since the minimum cost is 1 (the acknowledgment cost), the first property follows. To observe the upper bound on this value, notice that one of the possible solutions is to acknowledge each packet individually, and thus $\OPT(I_c^i) \leq n$.
\end{proof}

Combining the above lemmas gives us \cref{thm:nested}. Indeed, by \cref{lem:descendant_levels} and a recursive application of \cref{cor:parent_cost_equal_children}, we have that the sum of optimal costs for all the critical batches processed by our algorithm (i.e., those for which a service was issued) is at most equal to $\log n \cdot \sum_{i=1}^\ell \texttt{OPT}(I_c^i)$. Thus, by \cref{lem:critical_set_service_cost}, the algorithm paid at most $14 \log n \cdot \sum_{i=1}^\ell \texttt{OPT}(I_c^i)$ to serve all the packets. The lower bound on optimal cost from \cref{cor:total_cost_vs_critical_sets} completes the proof.

\subsection{Lower Bounds for Arbitrary Algorithms}

We %
show that the deterministic $O(\log n)$ competitive ratio of our algorithm (\cref{alg:tcp_sum_monotone_delay}) is tight.
We start by defining the Parking Permit Problem \cite{Meyerson05}. It is
known that for every
deterministic algorithm $\B$ for Parking Permit, there is an
adversarial request sequence $A$ on which $\B$ is $\Omega(\log n)$-competitive
(\cref{lem:summono-lb:ppadv}~\cite{NaorUW22}). Then, we show how
to convert any deterministic algorithm $\A$ for TCP Acknowledgment into a
deterministic algorithm $\BA$ for Parking Permit. Finally, we show that $\A$ is
$\Omega(\log n)$-competitive on the adversarial sequence for $\BA$ (\cref{lem:summono-lb:lb}).

\paragraph{Parking Permit Problem \cite{Meyerson05}.} An instance of the Parking
Permit Problem. For each timestep $t \geq 1$, a request may arrive, and if so, needs
to be covered by a permit. There are several different types of permits that can be
purchased at any time throughout the time horizon. A permit of type $k \geq 0$ has cost
$C_k$ and duration $D_k$: when bought at time $t$, it covers the time interval
$[t,t+D_k]$.
The goal is to minimize the cost of the purchased permits.

\paragraph{Hard Parking Permit instances.} A \emph{hard Parking Permit instance}
is one where $C_k=2^k$ and $D_k=4^k$ for each $k \geq 0$.
\begin{lemma}[\cite{NaorUW22}]
  \label{lem:summono-lb:ppadv}
  Let $\B$ be a deterministic online algorithm for Parking Permit. Consider
  a hard Parking Permit instance where
  the sequence of request arrival times $A$ is generated by the following
  adversary: for each $1 \leq j \leq n$ (in ascending order), request $j$
  arrives at the earliest timestep that is not yet covered by the permits bought
  by $\B$.

  When run on $A$, algorithm $\B$
  \begin{itemize}
  \item buys a separate permit $[a_j,t_j]$ for each
    request $j$; moreover, $a_{j+1} = t_j + 1$ for each $j < n$;
  \item has competitive ratio $\Omega(\log n)$.\footnote{We remark that the Parking Permit Problem is typically defined as having
a time horizon $[1,N]$ (on which $n$ requests arrive) and Lemma 21 of
\cite{NaorUW22} gives a lower bound of $\Omega(\log
N)$. Since $N \geq n$, we
get a lower bound of $\Omega(\log n)$.
}
  \end{itemize}
\end{lemma}

\paragraph{Hard TCP Acknowledgment instances.}
We now define the corresponding TCP Acknowledgment instances. Define the piecewise-linear concave function $f(x)=\min_{k \geq 0} C_k + x \cdot
(C_k/D_k)$ where $C_k$ and $D_k$ are defined as above.
A \emph{hard TCP Acknowledgment instance} is one in which the total cost of a
batch $\sigma$ acknowledged at time $t$---i.e., sum of the
acknowledgment cost\footnote{Recall that the acknowledgment cost is 1.} and the delay cost of the batch $\bdelay(\sigma,t)$---is
given by $f(t - a(\sigma))$ where $a(\sigma)$ is the earliest arrival time among
the packets in $\sigma$. That is, $\bdelay(\sigma,t) = f(t-a(\sigma)) - 1$.

\paragraph{Converting TCP Acknowledgment algorithm to Parking Permit.} Observe
that a hard Parking Permit instance with arrival sequence $A$ has a
corresponding hard TCP Acknowledgment instance with the same arrival sequence.
Our goal now is to show how to map decisions of a TCP Acknowledgment algorithm
on the corresponding instance into decisions for the
Parking Permit instance. We begin by presenting the two necessary ingredients.

First, we need to decide which permit to buy when a request arrives even though the TCP Acknowledgment
algorithm is not required to commit on arrival of a request, when the request
will be acknowledged.
For a sequence $A$, let $A_j$ denote the subsequence containing the first $j$
elements. The following lemma, together with \cref{prop:summono-lb:roundup},
lets us decide which permit to buy on arrival of a request.
\begin{lemma}
  \label{lem:summono-lb:threshold}
  Let $\A$ be a deterministic online algorithm for TCP Acknowledgment. Then,
  there exists a function $\nextack$ such that for every sequence of packet arrival
  times $A = (a_1, \ldots, a_n)$ and for every $j$, if $a_{j+1} > \nextack(A_j)$,
  then packet $j$ is acknowledged at time $\nextack(A)$. The function $\nextack$ is
  called the \emph{threshold function} of $\A$.
\end{lemma}

\begin{proof}
  Fix a sequence $A$ and $1 \leq j \leq n$. Define $\nextack(A_j) \geq a_j$ to be the
  time that $\A$ acknowledges the last packet of $A_j$ when $\A$ is run on
  $A_j$. If $a_{j+1} > \nextack(A_j)$, then up till time $\nextack(A_j)$, the
  executions of $\A$ on $A_j$ and of $\A$ on $A$ see the same information. Since $\A$
  is an online algorithm, the execution of $\A$ on $A$ will also acknowledge at
  time $\nextack(A_j)$.
\end{proof}

Second, we need to be able to map batches to permits. In general, the time between when a batch is acknowledged and the earliest arrival time
of the batch can be arbitrary, while permits can only have lengths $D_k = 4^k$.
The following lets us approximately map batches to permits.
\begin{proposition}
 \label{prop:summono-lb:roundup}
 For every $x \geq 0$, there exists $k^{*} \geq 0$ such that $D_{k^{*}} \geq x$ and $C_{k^{*}} \leq
 2f(x)$.
\end{proposition}

\begin{proof}
  Suppose $f(x) = C_k + x(C_k/D_k)$.
  If $x < D_k$, then $f(x) \geq C_k$ so choosing $k^{*} = k$ works.
  Next, we consider the case when $x > D_k$. Let $k^{*} > k$ be such that $D_{k^{*}-1} \leq x \leq
  D_{k^{*}}$. Observe that $C_{k}/D_{k} = (1/2)^{k} \leq (1/2)^{{k^{*}}-1} =
  C_{k^{*}-1}/D_{k^{*}-1}$. Together with the fact that $x \geq D_{k^{*}-1}$, we have
  \[f(x) \geq x(C_k/D_k) \geq
    D_{k^{*}-1}(C_{k^{*}-1}/D_{k^{*}-1}) = C_{k^{*}-1} = C_{k^{*}}/2,\]
  as desired.
\end{proof}

We now have the ingredients to convert any deterministic algorithm $\A$ for hard TCP Acknowledgment instances
 to a deterministic algorithm $\BA$ for hard Parking Permit
 instances. Let $\nextack(\cdot)$ be the threshold function of $\A$. The algorithm $\BA$ works as
 follows. For each request arrival $a_j$ that is not covered by permits bought
 so far:
 \begin{enumerate}[label=(\arabic*)]
 \item let $A_j$ be the sequence of arrival
   times so far, including $a_j$;
 \item let $k^{*}$ be such that $D_{k^{*}} \geq x$ and $f(D_{k^{*}})\leq 2 f(\nextack(A_j)-a_j)$ as guaranteed by
   \cref{prop:summono-lb:roundup};
 \item buy permit $[a_j, a_j+D_{k^{*}}]$.
 \end{enumerate}

 \begin{lemma}
   \label{lem:summono-lb:lb}
   Let $\A$ be a deterministic algorithm for hard TCP Acknowledgment instances
   and $A^{*}$ be the adversarial sequence for $\BA$ given by
   \cref{lem:summono-lb:ppadv}. Then, $\A$ is $\Omega(\log n)$-competitive on
   $A^*$.
 \end{lemma}

 \begin{proof}
 Let $\OPTPP(A^{*})$ and $\OPTTCP(A^{*})$ be the cost of the optimal Parking
 Permit and TCP Acknowledgment solutions to $A^{*}$, respectively.
 Since \cref{lem:summono-lb:ppadv} already shows that $\BA$ is $\Omega(\log
 n)$-competitive on $A^{*}$, it suffices to show the cost of $\A$ on $A^{*}$ is at least
 $\Omega(1)$ times the cost of $\BA$ on $A^{*}$, and that $\OPTTCP(A^{*}) =
 O(1) \OPTPP(A^{*})$.

 By \cref{lem:summono-lb:ppadv} and by definition of $\BA$, for every $j < n$,
 we have that $a_{j+1} > a_j+D_{{k^{*}}} \geq \nextack(A_j)$. Thus,
 for each request $j$, \cref{lem:summono-lb:threshold} implies that $\A$
 acknowledges $j$ by itself at time $\nextack(A_j)$ and the total cost of the batch is
 $f(\nextack(A_j)-a_j))$; on the other hand, $\BA$ buys permit
 $[a_j,a_j+D_{k^{*}}]$ of cost $C_{k^{*}} \leq
 2f(\nextack(A_j)-a_j)$. Thus, the cost of $\A$ on $A^{*}$ is at least 1/2
 the cost of $\BA$ on $A^{*}$.

 Next, we show that $\OPTTCP(A^{*}) =
 O(1) \OPTPP(A^{*})$. Let $P_i = [s_i,t_i]$ be the $i$-th permit bought by the optimal solution
for the hard instance of Parking Permit with arrival sequence $A$. We construct
a solution to the hard TCP Acknowledgment instance as follows: while there is an unacknowledged packet, let $a_j$ be the earliest
  unacknowledged packet and send an acknowledgment at time $t$ where
  $t$ is the furthest end time among all permits covering $a_j$ (at least one
  such permit exists since the permits have to cover $A$).

  We now argue that we can charge the total cost of each batch to a unique
  permit. For each batch
  $\sigma_i$ acknowledged at time $t_i$, let $P_{(i)}$ be the permit that covers
  $a(\sigma_i)$ and has end time $t_i$. Since $P_{(i)}$ is the permit with the
  furthest end time, we have that $P_{(i')} \neq P_{(i)}$ for $i' \neq i$. Let
  $C_{(i)}$ and $D_{(i)}$ be the cost and duration of $P_{(i)}$.
  Since $a(\sigma_i) \in P_{(i)}$, we get that the total cost of
  $\sigma_i$ is $f(t_i - a(\sigma_i)) \leq f(D_{(i)}) = 2C_{(i)}$. Thus the cost
  of each batch $\sigma_i$ is charged to at most twice the cost of permit
  $P_{(i)}$, and each permit is charged by at most one batch. Thus,
  we have shown that $\OPTTCP(A^{*}) = O(1)\OPTPP(A^{*})$, completing the proof
  of the lemma.
 \end{proof}

 Thus, we conclude that every deterministic algorithm for TCP Acknowledgment with sum-monotone delay costs is $\Omega(\log n)$-competitive, as desired. This also completes the proof of \cref{thm:summono}.

\section{Batch-Oblivious Delay Costs}

We now turn to the batch-oblivious model, where the delay cost is defined globally over the vector of packet delays, without reference to how packets are grouped in acknowledgments. Formally, the objective is to minimize $k + f(d)$, where $k$ is the number of acknowledgments, $f : \mathbb{R}_+^n \to \mathbb{R}_+$ is a monotone delay function, and $d = (d_1, \dots, d_n)$ is the vector of waiting times, where $d_i$ denotes how long packet $i$ waits before being acknowledged. The function $f$ is revealed incrementally as packets arrive. While this framework captures natural objectives such as total waiting time, it also includes more general functions, such as concave or submodular delay costs. As discussed earlier, greedy algorithms that succeed for additive delay fail in these settings, due to the ability of the optimal solution to concentrate delay on a few early packets and reduce marginal costs for the rest. A formal lower bound for concave delay functions is provided in \cref{sec:concave}. Here, we focus on broader classes of delay functions, including submodular functions and norms.

    \subsection{Concave Delay Cost}
    \label{sec:concave}
\concavelb*

\begin{proof}
 Let $\ell \leq n/2$ and $\epsilon\in [0,1]$ be parameters that will
 be set later. Define $x, y \in \mathbb{R}^n$ such that: each of the first $\ell$
 coordinates of $x$ is $\epsilon$ and each of the remaining coordinates is
 $1$; each of the first $\ell$ coordinates of $y$ is $n/\ell$ and  each of
 the remaining is $\epsilon$. Given a waiting time vector $d$, define $d_{\leq
   \ell} = \sum_{i \leq \ell} d_i$ and $d_{> \ell} = \sum_{i > \ell} d_i$. We now define the delay function
 \[
   f(d)
   = \min\{d \cdot x, d \cdot y\}
   = \min\left\{\epsilon d_{\leq \ell} + d_{> \ell},
   \frac{n}{\ell} d_{\leq \ell} + \epsilon d_{> \ell}\right\}.
 \]
 Note that the function $f$ is concave as it is the
 point-wise minimum of linear functions.

 The adversary proceeds as follows. For the first $\ell$ packets, it
 releases packet $i$ at time $i$. Let $A_\ell$ be the number of packets made by
 the algorithm just before time $\ell+1$. If $A_\ell \geq \ell/2$ (case
 1), it
 releases the remaining $n - \ell$ packets together at time $\ell + 1$;
 otherwise (case 2), it continues releasing each remaining packet $i$ at time $i$.

 We now show that in both cases, the adversary can construct a significantly cheaper solution. Consider case 1: the algorithm incurs a cost of at least $\ell/2$. Meanwhile, the adversary can use a single acknowledgment at time $\ell + 1$. Let $d^*$ denote the resulting waiting time vector, where $d^*_i = \ell + 1 - i$ for $i \leq \ell$, and $d^*_i = 0$ for $i > \ell$. The delay cost of this solution is at most $f(d^*) \leq \epsilon \sum_{i=1}^\ell i = O(\ell^2 \epsilon)$, and the total cost is $1 + O(\ell^2 \epsilon)$. Therefore, the competitive ratio of the algorithm is $\Omega\big(\min\{\ell,\, 1/(\ell \epsilon)\}\big)$ in this case.

Next, we consider case 2. Let $k$ denote the total number of acknowledgments made by the algorithm. Without loss of generality, we may assume that each acknowledgment occurs at the arrival time of some packet, as this can only reduce the total cost. Recall that in this case, packets arrive at unit times $1, 2, \dots, n$. Thus, for any time interval starting at a packet arrival, the number of acknowledgments made during the interval plus the total waiting time incurred by the packets arriving within it must be at least the length of the interval. This is because the acknowledgment cost is 1, so each acknowledgment can “pay for” at most one unit of waiting time.

 Now, consider the first $\ell$ packets. By assumption, the algorithm makes $A_\ell < \ell/2$ acknowledgments during this interval. From our previous observation, it follows that the total delay incurred by these packets is at least $\ell - A_\ell > \ell/2$. Thus, we obtain that $d \cdot y \geq (n/\ell) \cdot (\ell/2) = \Omega(n)$. On the other hand, consider only the remaining $n - \ell$ packets. Here, the number of acknowledgments and the overall waiting time sum up to $(k - A_\ell) + d_{>\ell}$. As noted above, this sum must be at least the length of the interval $(\ell, n]$, which is greater than $n/2$, since $\ell < n/2$. Hence, we have $k + d \cdot x \geq (k - A_\ell) + d_{>\ell} = \Omega(n)$. Thus, the total cost of the algorithm satisfies $k + \min\{d \cdot x, d \cdot y\} \geq \Omega(n)$.

 On the other hand, consider the solution that acknowledges each of the first
 $\ell$ packets on arrival and makes one acknowledgment at the end at time
 $n$. Let $d^{*}$ be its waiting time vector. Then, it holds that $f(d^{*}) \leq d^{*}\cdot
 y = \epsilon d_{>\ell} \leq O(n^2\epsilon)$ and so the solution has total cost
 $\ell + O(n^2\epsilon)$. Thus, the competitive ratio is
 $\Omega(\min\{n/\ell, 1/(n\epsilon)\})$ in this case.

 Setting $\ell = \lceil \sqrt{n} \rceil$ and $\epsilon = 1/n^2$ gives the theorem.
\end{proof}

\subsection{Submodular Delay Cost and Norms}
We begin by defining continuous submodularity. 
\begin{definition}[Continuous submodular functions]
  \label{def:cs}
    A function $f : \mathbb{R}^n_+ \rightarrow \mathbb{R}_+$ is \emph{continuous
      submodular} if for every $x,y \in \mathbb{R}^n$, it satisfies $f(x \vee y)
    + f(x \wedge y) \leq f(x) + f(y)$, where $\vee$ denotes coordinate-wise
    maximum and $\wedge$ denotes coordinate-wise minimum. 
\end{definition}

We remark that, unlike the case of discrete submodular functions, continuous
submodularity is only equivalent to a weaker notion of diminishing marginals: for
every $x,y \in \mathbb{R}^n_+$ such that $x \leq y$, $i \in [n]$ such that
$x_i = y_i$ and $a \geq 0$, we have $f(x + ae_i) - f(x) \geq f(y + ae_i) -
f(y)$, where $e_i$ is the $i$-th standard basis vector. We refer the reader to
\cite[Remark 7 and Appendix~A.1]{PattonR023} for further
discussion on this topic.

In fact, our analysis only uses an even weaker condition, where the
diminishing-marginals property is required only when the coordinate being
increased is zero in both vectors. We say that $f$ satisfies
\emph{zero-coordinate diminishing marginals} if for every $x,y\in\mathbb R^n_+$
such that $x\le y$, every $i\in[n]$ such that $x_i=y_i=0$, and every $a\ge 0$,
we have $f(x+ae_i)-f(x)\ge f(y+ae_i)-f(y)$. Clearly, continuous submodularity
implies this condition.

For a vector $\Delta\in\mathbb R^n_+$, let
$\operatorname{supp}(\Delta)=\{i\in[n]:\Delta_i>0\}$. We will use the following
simple observation.

\begin{observation}
  \label{obs:zero-coordinate-vector-marginal}
  Suppose that $f$ satisfies zero-coordinate diminishing marginals. Let
  $u,v,\Delta\in\mathbb R^n_+$ be such that $u\le v$ and $u_i=v_i=0$ for every
  $i\in \operatorname{supp}(\Delta)$. Then
  \[
    f(u+\Delta)-f(u)\ge f(v+\Delta)-f(v).
  \]
\end{observation}

Indeed, one can add the coordinates in $\operatorname{supp}(\Delta)$ one at a
time. Before each coordinate is added, it is still zero in both the smaller and
the larger vector, so zero-coordinate diminishing marginals apply. Summing
these inequalities over all coordinates in $\operatorname{supp}(\Delta)$ gives the
desired inequality.

\submodular*

\begin{proof} 
Recall that Greedy makes an acknowledgment
whenever its delay cost increases by $1$. We now analyze the competitive ratio
of the algorithm. First, observe that by construction, the cost of the algorithm is at most twice
the number of acknowledgments. Let $k$ be the number of acknowledgments made by
the algorithm. Fix an optimal solution. In the rest of the proof, we argue
that its cost is at least $k$.

Let $t_j$ be the time when the algorithm made its $j$-th acknowledgment, and define
$T_j = (t_{j-1}, t_j]$ with $t_0 = 0$. For each packet $i$, let $d_i(j)$
and $d^*_i(j)$ be the amount of waiting time that packet $i$ incurred up till
$t_j$ in the algorithm's solution and the optimal solution, respectively. If
packet $i$ arrived after $t_j$, we define $d_i(j)$ and $d^*_i(j)$ to be both $0$.

At a high level, we adopt the proof framework of \cite{dooly2001line} for the
standard TCP Acknowledgment setting and generalize it to our case: within each
interval $T_j$, the optimal solution either issues an acknowledgment and pays a
cost of $1$ or its delay cost increases by at least $1$. However, if the optimal
solution allows earlier packets to wait a long time, it can effectively
``front-load'' its delay cost, which causes the marginal cost in later intervals
to be too small. To address this, we introduce a capped version of the optimal
delay vector, denoted $\capd$, where each entry is defined as
$\capd_i(j) = \min\{d^*_i(j), d_i(j)\}$. In what follows, we use $d(j)$ and
$\capd(j)$ to denote the vectors $(d_1(j), \dots, d_n(j))$ and
$(\capd_1(j), \dots, \capd_n(j))$, respectively. By monotonicity, the optimal
solution's delay cost $f(d^*) \geq f(\capd)$ and so it suffices to show that if
the optimal solution makes no acknowledgment in $T_j$, then the increase in the
delay cost of the capped optimal delay vector in $T_j$ is
$f(\capd(j)) - f(\capd(j-1)) \geq 1$. 

\begin{claim}
  \label{clm:CS-marginal}
    For each interval $T_j$, if the optimal solution did not make an acknowledgment in $T_j$, then 
    \[
      f(\capd(j)) - f(\capd(j-1)) \geq f(d(j)) - f(d(j-1)) = 1.
    \]
\end{claim}

\begin{proof}
  The equality follows from the fact that $f(d(j)) - f(d(j-1))$ is the increase
  in the algorithm's delay cost during interval $T_j$, and the algorithm makes an
  acknowledgment exactly when its delay cost increases by $1$ since the last
  acknowledgment. 
  
  It remains to prove the inequality. Let $\Delta^j = d(j)-d(j-1)$ be the
  increase in the algorithm's delay vector during interval $T_j$. Since the
  algorithm acknowledged all pending packets at time $t_{j-1}$, the only packets
  whose delay can increase during $T_j$ are the packets that arrive
  during $T_j$. Hence, for every $i\in\operatorname{supp}(\Delta^j)$, we have
  $d_i(j-1)=0$. Moreover, by the definition of the capped vector,
  $\capd_i(j-1)=0$. We also have $\capd(j-1)\le d(j-1)$. Therefore, applying
  \cref{obs:zero-coordinate-vector-marginal} with
  $u=\capd(j-1)$, $v=d(j-1)$, and $\Delta=\Delta^j$, we get
  \[
    f(\capd(j-1)+\Delta^j)-f(\capd(j-1))
    \ge
    f(d(j-1)+\Delta^j)-f(d(j-1)).
  \]
  Since $d(j-1)+\Delta^j=d(j)$, this gives
  \begin{equation}
  \label{eq:cs}
    f(\capd(j-1)+\Delta^j)-f(\capd(j-1))
    \ge
    f(d(j))-f(d(j-1)).
  \end{equation}

  It remains to compare $\capd(j-1)+\Delta^j$ with $\capd(j)$. We claim that
  $\capd(j-1)+\Delta^j \leq \capd(j)$ coordinatewise. Fix a packet $i$. If $i$
  arrived during $T_j$, then, since the optimal solution did not make an
  acknowledgment in $T_j$, the packet has the same waiting time up till $t_j$ in
  both solutions. Thus, $d_i^*(j)=d_i(j)$, and hence
  $\capd_i(j)=d_i(j)=\Delta_i^j$. Also, $\capd_i(j-1)=0$, so the desired
  inequality holds with equality.

  If packet $i$ arrived no later than $t_{j-1}$, then the algorithm does not
  incur additional waiting time for packet $i$ during $T_j$, so
  $\Delta_i^j=0$. Since waiting times in the optimal solution do not decrease
  while no acknowledgment is made, and since $d_i(j)=d_i(j-1)$, we have
  $\capd_i(j-1)\le \capd_i(j)$. Thus, the desired inequality again holds.
  Finally, if packet $i$ arrives after $t_j$, all relevant quantities are $0$,
  and the inequality is trivial. Therefore, $\capd(j-1)+\Delta^j \leq \capd(j)$, as desired. 
  
  Applying monotonicity of $f$ to the above inequality, we have
  $f(\capd(j)) \geq f(\capd(j-1)+\Delta^j)$. Plugging this into \cref{eq:cs} gives
  \[
    f(\capd(j))-f(\capd(j-1))
    \ge
    f(d(j))-f(d(j-1)).
  \]
  Together with the greedy rule $f(d(j))-f(d(j-1))=1$, this proves the claim.
\end{proof}

Suppose the optimal solution makes $k^*$ acknowledgments. If $k^* \geq k$, then
we are done. Otherwise, there exist at least $k - k^*$ intervals in which the
optimal solution did not make an acknowledgment. By monotonicity of $f$ and \cref{clm:CS-marginal},
the optimal solution's delay cost is 
\[
    f(d^*) 
    \geq f(\capd) 
    = \sum_{j=1}^k f(\capd(j)) - f(\capd(j-1)) 
    \geq k- k^*.
\]  
Thus, the optimal solution's total cost is at least $k$. 
\end{proof}

Patton, Russo, and Singla~\cite{PattonR023} observed that $\ell_p$ norms, $\textsf{Top-}k$ norms and ordered norms are
submodular. They also proved a result on the approximability of symmetric norms by
submodular norms. Using these results give us the following corollary.

\normthm*

\begin{remark}
    The zero-coordinate diminishing-marginals condition has a natural modeling
interpretation in practical domains beyond TCP Acknowledgment. It says that the marginal cost of making a new coordinate
positive is smaller when more other coordinates are already present. This is
appropriate for incident-based or SLA-type\footnote{SLA = Service Level Agreement.} delay costs, where the first few
delayed requests may create most of the operational or reputational cost, while
additional newly delayed requests are partly absorbed into the same incident. For
example, functions of the form $f(d)=g(\operatorname{supp}(d))$, where $g$ is a
monotone submodular function, satisfy zero-coordinate diminishing marginals:
adding a new positive coordinate corresponds exactly to adding a new element to a
larger set, whose marginal contribution is smaller by the submodularity of $g$. Such
functions model global costs depending on the set of requests that experience any
positive delay, rather than on the precise magnitude of each delay.
\end{remark}

\section{Conclusion}

In this work, we revisit Online TCP Acknowledgment under broad and natural generalizations of delay costs. We show that the classical guarantees of Greedy extend well beyond additive delays: it remains $2$-competitive for \emph{max-monotone} objectives and for a rich class of \emph{batch-oblivious} continuous submodular delay
costs. At the same time, we identify a sharp boundary where this paradigm
breaks, by proving that Greedy can be $\Omega(n)$-competitive for \emph{sum-monotone}
delays. For this harder regime, we develop a new phase-based online algorithm
that achieves a tight $\Theta(\log n)$ competitive ratio. Together, our results
give a more complete picture of TCP under non-additive delay objectives and lay
conceptual groundwork for online problems with delay more broadly. 

There are several intriguing open problems. First, what is the competitiveness of other problems such as Joint Replenishment, Multi-Level Aggregation, and Set Cover with Delay under the more general delay cost models considered in this work? Second, while we have a tight characterization of the deterministic competitive ratio for sum-monotone delays, it is unclear whether randomization can help. For the Parking Permit problem, Meyerson~\cite{Meyerson05} showed that the randomized competitive ratio is $\Theta(\log \log n)$. A reduction from Parking Permit will imply the same lower bound. Finally, are there other interesting delay objectives? For example, fairness-oriented delay objectives, where the goal is
to balance efficiency with equitable treatment of packets, a setting that
remains largely unexplored in the TCP acknowledgment framework.

\small
\bibliography{pubs}

@inproceedings{dooly1998tcp,
  title={Tcp dynamic acknowledgment delay (extended abstract) theory and practice},
  author={Dooly, Daniel R and Goldman, Sally A and Scott, Stephen D},
  booktitle={Proceedings of the thirtieth annual ACM symposium on Theory of computing},
  pages={389--398},
  year={1998}
}

@article{KarlinKR01,
	title = {Dynamic {TCP} {Acknowledgment} and {Other} {Stories} about e/(e - 1)},
	volume = {36},
	issn = {1432-0541},
	url = {https://doi.org/10.1007/s00453-003-1013-x},
	doi = {10.1007/s00453-003-1013-x},
	number = {3},
	journal = {Algorithmica},
	author = {{Karlin} and {Kenyon} and {Randall}},
	month = jul,
	year = {2003},
	pages = {209--224},
}

@article{dooly2001line,
  title={On-line analysis of the TCP acknowledgment delay problem},
  author={Dooly, Daniel R and Goldman, Sally A and Scott, Stephen D},
  journal={Journal of the ACM (JACM)},
  volume={48},
  number={2},
  pages={243--273},
  year={2001},
  publisher={ACM New York, NY, USA}
}

@misc{clark1982rfc0813,
  title={RFC0813: Window and Acknowledgement Strategy in TCP},
  author={Clark, David D},
  year={1982},
  publisher={RFC Editor}
}

@article{allman1998generation,
  title={On the generation and use of TCP acknowledgments},
  author={Allman, Mark},
  journal={ACM SIGCOMM Computer Communication Review},
  volume={28},
  number={5},
  pages={4--21},
  year={1998},
  publisher={ACM New York, NY, USA}
}

@inproceedings{Seiden00,
  author       = {Steven S. Seiden},
  title        = {A guessing game and randomized online algorithms},
  booktitle    = {Proceedings of the Thirty-Second Annual {ACM} Symposium on Theory
                  of Computing, May 21-23, 2000, Portland, OR, {USA}},
  pages        = {592--601},
  publisher    = {{ACM}},
  year         = {2000}
}

@inproceedings{de2005dynamic,
  title={A dynamic adaptive acknowledgment strategy for TCP over multihop wireless networks},
  author={De Oliveira, Ruy and Braun, Torsten},
  booktitle={Proceedings IEEE 24th Annual Joint Conference of the IEEE Computer and Communications Societies.},
  volume={3},
  pages={1863--1874},
  year={2005},
  organization={IEEE}
}

@inproceedings{AlbersB03,
  author       = {Susanne Albers and
                  Helge Bals},
  title        = {Dynamic {TCP} acknowledgement: penalizing long delays},
  booktitle    = {Proceedings of the Fourteenth Annual {ACM-SIAM} Symposium on Discrete
                  Algorithms, January 12-14, 2003, Baltimore, Maryland, {USA}},
  pages        = {47--55},
  publisher    = {{ACM/SIAM}},
  year         = {2003}
}

@article{albers2003online,
  title={Online algorithms: a survey},
  author={Albers, Susanne},
  journal={Mathematical Programming},
  volume={97},
  pages={3--26},
  year={2003},
  publisher={Springer}
}

@article{AlbersB05,
  author       = {Susanne Albers and
                  Helge Bals},
  title        = {Dynamic {TCP} Acknowledgment: Penalizing Long Delays},
  journal      = {{SIAM} J. Discret. Math.},
  volume       = {19},
  number       = {4},
  pages        = {938--951},
  year         = {2005}
}

@article{bienkowski2020online,
  title={Online algorithms for multilevel aggregation},
  author={Bienkowski, Marcin and B{\"o}hm, Martin and Byrka, Jaroslaw and Chrobak, Marek and D{\"u}rr, Christoph and Folwarczn{\`y}, Luk{\'a}{\v{s}} and Je{\.z}, {\L}ukasz and Sgall, Ji{\v{r}}{\'\i} and Thang, Nguyen Kim and Vesel{\`y}, Pavel},
  journal={Operations Research},
  volume={68},
  number={1},
  pages={214--232},
  year={2020},
  publisher={INFORMS}
}

@inproceedings{le2023power,
  title={The power of clairvoyance for multi-level aggregation and set cover with delay},
  author={Le, Ngoc Mai and Umboh, Seeun William and Xie, Ningyuan},
  booktitle={Proceedings of the 2023 Annual ACM-SIAM Symposium on Discrete Algorithms (SODA)},
  pages={1594--1610},
  year={2023},
  organization={SIAM}
}

@inproceedings{chen2022online,
  author       = {Ryder Chen and
                  Jahanvi Khatkar and
                  Seeun William Umboh},
  editor       = {Mikolaj Bojanczyk and
                  Emanuela Merelli and
                  David P. Woodruff},
  title        = {Online Weighted Cardinality Joint Replenishment Problem with Delay},
  booktitle    = {49th International Colloquium on Automata, Languages, and Programming,
                  {ICALP} 2022, Paris, France, July 4-8, 2022},
  series       = {LIPIcs},
  volume       = {229},
  pages        = {40:1--40:18},
  publisher    = {Schloss Dagstuhl - Leibniz-Zentrum f{\"{u}}r Informatik},
  year         = {2022},
  url          = {https://doi.org/10.4230/LIPIcs.ICALP.2022.40},
  doi          = {10.4230/LIPICS.ICALP.2022.40},
  bibsource    = {dblp computer science bibliography, https://dblp.org}
}

@inproceedings{emek2016online,
  title={Online matching: haste makes waste!},
  author={Emek, Yuval and Kutten, Shay and Wattenhofer, Roger},
  booktitle={Proceedings of the forty-eighth annual ACM symposium on Theory of Computing},
  pages={333--344},
  year={2016}
}

@article{ashlagi2017min,
  title={Min-cost bipartite perfect matching with delays},
  author={Ashlagi, Itai and Azar, Yossi and Charikar, Moses and Chiplunkar, Ashish and Geri, Ofir and Kaplan, Haim and Makhijani, Rahul and Wang, Yuyi and Wattenhofer, Roger},
  journal={Approximation, Randomization, and Combinatorial Optimization. Algorithms and Techniques (APPROX/RANDOM 2017)},
  volume={81},
  pages={1},
  year={2017},
  publisher={Schloss Dagstuhl-Leibniz-Zentrum f{\"u}r Informatik}
}

@inproceedings{azar2021min,
  title={The min-cost matching with concave delays problem},
  author={Azar, Yossi and Ren, Runtian and Vainstein, Danny},
  booktitle={Proceedings of the 2021 ACM-SIAM Symposium on Discrete Algorithms (SODA)},
  pages={301--320},
  year={2021},
  organization={SIAM}
}

@inproceedings{bienkowski2017match,
  title={A match in time saves nine: Deterministic online matching with delays},
  author={Bienkowski, Marcin and Kraska, Artur and Schmidt, Pawe{\l}},
  booktitle={International Workshop on Approximation and Online Algorithms},
  pages={132--146},
  year={2017},
  organization={Springer}
}

@inproceedings{azar2017online,
  title={Online service with delay},
  author={Azar, Yossi and Ganesh, Arun and Ge, Rong and Panigrahi, Debmalya},
  booktitle={Proceedings of the 49th Annual ACM SIGACT Symposium on Theory of Computing},
  pages={551--563},
  year={2017}
}

@inproceedings{azar2020beyond,
  title={Beyond tree embeddings--a deterministic framework for network design with deadlines or delay},
  author={Azar, Yossi and Touitou, Noam},
  booktitle={2020 IEEE 61st Annual Symposium on Foundations of Computer Science (FOCS)},
  pages={1368--1379},
  year={2020},
  organization={IEEE}
}

@inproceedings{AzarCKT20,
  author       = {Yossi Azar and
                  Ashish Chiplunkar and
                  Shay Kutten and
                  Noam Touitou},
  title        = {Set Cover with Delay - Clairvoyance Is Not Required},
  booktitle    = {28th Annual European Symposium on Algorithms, {ESA} 2020, September
                  7-9, 2020, Pisa, Italy (Virtual Conference)},
  series       = {LIPIcs},
  volume       = {173},
  pages        = {8:1--8:21},
  publisher    = {Schloss Dagstuhl - Leibniz-Zentrum f{\"{u}}r Informatik},
  year         = {2020}
}

@inproceedings{PattonR023,
  author       = {Kalen Patton and
                  Matteo Russo and
                  Sahil Singla},
  title        = {Submodular Norms with Applications To Online Facility Location and
                  Stochastic Probing},
  booktitle    = {Approximation, Randomization, and Combinatorial Optimization. Algorithms
                  and Techniques, {APPROX/RANDOM} 2023, September 11-13, 2023, Atlanta,
                  Georgia, {USA}},
  series       = {LIPIcs},
  volume       = {275},
  pages        = {23:1--23:22},
  publisher    = {Schloss Dagstuhl - Leibniz-Zentrum f{\"{u}}r Informatik},
  year         = {2023}
}

@inproceedings{CarrascoPSV18,
  author       = {Rodrigo A. Carrasco and
                  Kirk Pruhs and
                  Cliff Stein and
                  Jos{\'{e}} Verschae},
  title        = {The Online Set Aggregation Problem},
  booktitle    = {{LATIN} 2018: Theoretical Informatics - 13th Latin American Symposium,
                  Buenos Aires, Argentina, April 16-19, 2018, Proceedings},
  series       = {Lecture Notes in Computer Science},
  volume       = {10807},
  pages        = {245--259},
  publisher    = {Springer},
  year         = {2018}
}

@inproceedings{moseley2025putting,
  title={Putting Off the Catching Up: Online Joint Replenishment Problem with Holding and Backlog Costs},
  author={Moseley, Benjamin and Niaparast, Aidin and Ravi, R},
  booktitle={Proceedings of the 2025 Annual ACM-SIAM Symposium on Discrete Algorithms (SODA)},
  pages={3865--3883},
  year={2025},
  organization={SIAM}
}

@article{gyorgyi2023joint,
  title={Joint replenishment meets scheduling},
  author={Gy{\"o}rgyi, P{\'e}ter and Kis, Tam{\'a}s and Tam{\'a}si, T{\'\i}mea and B{\'e}k{\'e}si, J{\'o}zsef},
  journal={Journal of Scheduling},
  volume={26},
  number={1},
  pages={77--94},
  year={2023},
  publisher={Springer}
}

@article{BuchbinderKLMS13,
  author       = {Niv Buchbinder and
                  Tracy Kimbrel and
                  Retsef Levi and
                  Konstantin Makarychev and
                  Maxim Sviridenko},
  title        = {Online Make-to-Order Joint Replenishment Model: Primal-Dual Competitive
                  Algorithms},
  journal      = {Oper. Res.},
  volume       = {61},
  number       = {4},
  pages        = {1014--1029},
  year         = {2013},
  url          = {https://doi.org/10.1287/opre.2013.1188},
  doi          = {10.1287/OPRE.2013.1188},
  bibsource    = {dblp computer science bibliography, https://dblp.org}
}

@inproceedings{BienkowskiBCJS13,
  author    = {Marcin Bienkowski and
               Jaroslaw Byrka and
               Marek Chrobak and
               Lukasz Jez and
               Dorian Nogneng and
               Jir{\'{\i}} Sgall},
  title     = {Better Approximation Bounds for the Joint Replenishment Problem},
  booktitle = {Proceedings of the Twenty-Fifth Annual {ACM-SIAM} Symposium on Discrete
               Algorithms, {SODA} 2014, Portland, Oregon, USA, January 5-7, 2014},
  pages     = {42--54},
  publisher = {{SIAM}},
  year      = {2014},
  url       = {https://doi.org/10.1137/1.9781611973402.4},
  doi       = {10.1137/1.9781611973402.4},
  bibsource = {dblp computer science bibliography, https://dblp.org}
}

@inproceedings{BritoKV04,
  author       = {Carlos Brito and
                  Elias Koutsoupias and
                  Shailesh Vaya},
  title        = {Competitive analysis of organization networks or multicast acknowledgement:
                  how much to wait?},
  booktitle    = {Proceedings of the Fifteenth Annual {ACM-SIAM} Symposium on Discrete
                  Algorithms, {SODA} 2004, New Orleans, Louisiana, USA, January 11-14,
                  2004},
  pages        = {627--635},
  publisher    = {{SIAM}},
  year         = {2004},
  url          = {http://dl.acm.org/citation.cfm?id=982792.982886},
  bibsource    = {dblp computer science bibliography, https://dblp.org}
}

@inproceedings{frameworkDelay,
  author    = {Yossi Azar and
               Noam Touitou},
  title     = {General Framework for Metric Optimization Problems with Delay or with
               Deadlines},
  booktitle = {Proceedings of the 60th {IEEE} Annual Symposium on Foundations of Computer Science, {FOCS}
               2019},
  pages     = {60--71},
  publisher = {{IEEE} Computer Society},
  year      = {2019},
  url       = {https://doi.org/10.1109/FOCS.2019.00013},
  doi       = {10.1109/FOCS.2019.00013},
  bibsource = {dblp computer science bibliography, https://dblp.org}
}

@inproceedings{Touitou21,
  author    = {Noam Touitou},
  title     = {Nearly-Tight Lower Bounds for Set Cover and Network Design with Deadlines/Delay},
  booktitle = {32nd International Symposium on Algorithms and Computation, {ISAAC}
               2021, December 6-8, 2021, Fukuoka, Japan},
  series    = {LIPIcs},
  volume    = {212},
  pages     = {53:1--53:16},
  publisher = {Schloss Dagstuhl - Leibniz-Zentrum f{\"{u}}r Informatik},
  year      = {2021},
  url       = {https://doi.org/10.4230/LIPIcs.ISAAC.2021.53},
  doi       = {10.4230/LIPIcs.ISAAC.2021.53},
  bibsource = {dblp computer science bibliography, https://dblp.org}
}

@inproceedings{Touitou23,
  author       = {Noam Touitou},
  title        = {Frameworks for Nonclairvoyant Network Design with Deadlines or Delay},
  booktitle    = {50th International Colloquium on Automata, Languages, and Programming,
                  {ICALP} 2023, July 10-14, 2023, Paderborn, Germany},
  series       = {LIPIcs},
  volume       = {261},
  pages        = {105:1--105:20},
  publisher    = {Schloss Dagstuhl - Leibniz-Zentrum f{\"{u}}r Informatik},
  year         = {2023},
  url          = {https://doi.org/10.4230/LIPIcs.ICALP.2023.105},
  doi          = {10.4230/LIPICS.ICALP.2023.105},
  bibsource    = {dblp computer science bibliography, https://dblp.org}
}

@inproceedings{BuchbinderFNT17,
  author    = {Niv Buchbinder and
               Moran Feldman and
               Joseph (Seffi) Naor and
               Ohad Talmon},
  title     = {\emph{O}(depth)-Competitive Algorithm for Online Multi-level Aggregation},
  booktitle = {Proceedings of the Twenty-Eighth Annual {ACM-SIAM} Symposium on Discrete
               Algorithms, {SODA} 2017, Barcelona, Spain, Hotel Porta Fira, January
               16-19},
  pages     = {1235--1244},
  publisher = {{SIAM}},
  year      = {2017},
  url       = {https://doi.org/10.1137/1.9781611974782.80},
  doi       = {10.1137/1.9781611974782.80},
  bibsource = {dblp computer science bibliography, https://dblp.org}
}

@article{BienkowskiBBCDF21,
  author       = {Marcin Bienkowski and
                  Martin B{\"{o}}hm and
                  Jaroslaw Byrka and
                  Marek Chrobak and
                  Christoph D{\"{u}}rr and
                  Luk\'a\v{s} Folwarczn\'y and
                  Lukasz Jez and
                  Jir{\'{\i}} Sgall and
                  Kim Thang Nguyen and
                  Pavel Vesel{\'{y}}},
  title        = {New results on multi-level aggregation},
  journal      = {Theor. Comput. Sci.},
  volume       = {861},
  pages        = {133--143},
  year         = {2021},
  url          = {https://doi.org/10.1016/j.tcs.2021.02.016},
  doi          = {10.1016/j.tcs.2021.02.016},
  bibsource    = {dblp computer science bibliography, https://dblp.org}
}

@article{francis,
	author = {Bach, Francis},
	date = {2019/05/01},
	doi = {10.1007/s10107-018-1248-6},
	id = {Bach2019},
	isbn = {1436-4646},
	journal = {Mathematical Programming},
	number = {1},
	pages = {419--459},
	title = {Submodular functions: from discrete to continuous domains},
	url = {https://doi.org/10.1007/s10107-018-1248-6},
	volume = {175},
	year = {2019}}

@inproceedings{DeryckereU23,
  author       = {Lindsey Deryckere and
                  Seeun William Umboh},
  title        = {Online Matching with Set and Concave Delays},
  booktitle    = {Approximation, Randomization, and Combinatorial Optimization. Algorithms
                  and Techniques, {APPROX/RANDOM} 2023, September 11-13, 2023, Atlanta,
                  Georgia, {USA}},
  series       = {LIPIcs},
  volume       = {275},
  pages        = {17:1--17:17},
  publisher    = {Schloss Dagstuhl - Leibniz-Zentrum f{\"{u}}r Informatik},
  year         = {2023},
  url          = {https://doi.org/10.4230/LIPIcs.APPROX/RANDOM.2023.17},
  doi          = {10.4230/LIPICS.APPROX/RANDOM.2023.17},
  bibsource    = {dblp computer science bibliography, https://dblp.org}
}

@inproceedings{ShmoysSU26,
  author       = {David Shmoys and
                  Varun Suriyanarayana and
                  Seeun William Umboh},
  title        = {Improved Online Algorithms for Inventory Management Problems with Holding and Delay Costs: Riding the Wave Makes Things Simpler, Stronger, \& More General},
  booktitle    = {Proceedings of the 2026 {ACM-SIAM} Symposium on Discrete Algorithms,
                  {SODA} 2026, Vancouver, Canada, January 11-14, 2026},
  pages        = {2726--2742},
  publisher    = {{SIAM}},
  year         = {2026},
  url          = {https://doi.org/10.1137/1.9781611978971.100},
  doi          = {10.1137/1.9781611978971.100}
}

@inproceedings{AzarL26,
  author       = {Yossi Azar and
                  Shahar Lewkowicz},
  title        = {Online Joint Replenishment Problem with Arbitrary Holding and Backlog Costs},
  booktitle    = {Proceedings of the 2026 {ACM-SIAM} Symposium on Discrete Algorithms,
                  {SODA} 2026, Vancouver, Canada, January 11-14, 2026},
  pages        = {2702--2725},
  publisher    = {{SIAM}},
  year         = {2026},
  url          = {https://doi.org/10.1137/1.9781611978971.99},
  doi          = {10.1137/1.9781611978971.99}
}

@inproceedings{Meyerson05,
  author       = {Adam Meyerson},
  title        = {The Parking Permit Problem},
  booktitle    = {46th Annual {IEEE} Symposium on Foundations of Computer Science, {FOCS}
                  2005, Pittsburgh, PA, USA, October 23-25, 2005, Proceedings},
  pages        = {274--284},
  publisher    = {{IEEE} Computer Society},
  year         = {2005},
  url          = {https://doi.org/10.1109/SFCS.2005.72},
  doi          = {10.1109/SFCS.2005.72},
  bibsource    = {dblp computer science bibliography, https://dblp.org}
}

@article{NaorUW22,
  author       = {Joseph (Seffi) Naor and
                  Seeun William Umboh and
                  David P. Williamson},
  title        = {Tight Bounds for Online Weighted Tree Augmentation},
  journal      = {Algorithmica},
  volume       = {84},
  number       = {2},
  pages        = {304--324},
  year         = {2022},
  url          = {https://doi.org/10.1007/s00453-021-00888-7},
  doi          = {10.1007/S00453-021-00888-7},
  bibsource    = {dblp computer science bibliography, https://dblp.org}
}

@inproceedings{AzarCK17,
  author       = {Yossi Azar and
                  Ashish Chiplunkar and
                  Haim Kaplan},
  editor       = {Philip N. Klein},
  title        = {Polylogarithmic Bounds on the Competitiveness of Min-cost Perfect
                  Matching with Delays},
  booktitle    = {Proceedings of the Twenty-Eighth Annual {ACM-SIAM} Symposium on Discrete
                  Algorithms, {SODA} 2017, Barcelona, Spain, Hotel Porta Fira, January
                  16-19},
  pages        = {1051--1061},
  publisher    = {{SIAM}},
  year         = {2017},
  url          = {https://doi.org/10.1137/1.9781611974782.67},
  doi          = {10.1137/1.9781611974782.67},
  bibsource    = {dblp computer science bibliography, https://dblp.org}
}

@article{KarlinMRS88,
  author       = {Anna R. Karlin and
                  Mark S. Manasse and
                  Larry Rudolph and
                  Daniel Dominic Sleator},
  title        = {Competitive Snoopy Caching},
  journal      = {Algorithmica},
  volume       = {3},
  pages        = {77--119},
  year         = {1988},
  url          = {https://doi.org/10.1007/BF01762111},
  doi          = {10.1007/BF01762111},
  bibsource    = {dblp computer science bibliography, https://dblp.org}
}

@article{DBLP:journals/algorithmica/EzraLPRU26,
  author       = {Tomer Ezra and
                  Stefano Leonardi and
                  Michal Pawlowski and
                  Matteo Russo and
                  Seeun William Umboh},
  title        = {Universal Optimization for Non-Clairvoyant Subadditive Joint Replenishment},
  journal      = {Algorithmica},
  volume       = {88},
  number       = {3},
  pages        = {43},
  year         = {2026},
  url          = {https://doi.org/10.1007/s00453-026-01389-1},
  doi          = {10.1007/S00453-026-01389-1},
  bibsource    = {dblp computer science bibliography, https://dblp.org}
}

@inproceedings{DBLP:conf/stacs/Touitou25,
  author       = {Noam Touitou},
  editor       = {Olaf Beyersdorff and
                  Michal Pilipczuk and
                  Elaine Pimentel and
                  Kim Thang Nguyen},
  title        = {Nearly-Optimal Algorithm for Non-Clairvoyant Service with Delay},
  booktitle    = {42nd International Symposium on Theoretical Aspects of Computer Science,
                  {STACS} 2025, Jena, Germany, March 4-7, 2025},
  series       = {LIPIcs},
  pages        = {74:1--74:21},
  publisher    = {Schloss Dagstuhl - Leibniz-Zentrum f{\"{u}}r Informatik},
  year         = {2025},
  url          = {https://doi.org/10.4230/LIPIcs.STACS.2025.74},
  doi          = {10.4230/LIPICS.STACS.2025.74},
  bibsource    = {dblp computer science bibliography, https://dblp.org}
}

@inproceedings{DBLP:conf/stoc/Touitou23,
  author       = {Noam Touitou},
  editor       = {Barna Saha and
                  Rocco A. Servedio},
  title        = {Improved and Deterministic Online Service with Deadlines or Delay},
  booktitle    = {Proceedings of the 55th Annual {ACM} Symposium on Theory of Computing,
                  {STOC} 2023, Orlando, FL, USA, June 20-23, 2023},
  pages        = {761--774},
  publisher    = {{ACM}},
  year         = {2023},
  url          = {https://doi.org/10.1145/3564246.3585107},
  doi          = {10.1145/3564246.3585107},
  bibsource    = {dblp computer science bibliography, https://dblp.org}
}

@inproceedings{DBLP:conf/soda/DufayW26,
  author       = {Marc Dufay and
                  Roger Wattenhofer},
  editor       = {Kasper Green Larsen and
                  Barna Saha},
  title        = {A Deterministic Polylogarithmic Competitive Algorithm for Matching
                  with Delays},
  booktitle    = {Proceedings of the 2026 Annual {ACM-SIAM} Symposium on Discrete Algorithms,
                  {SODA} 2026, Vancouver, BC, Canada, January 11-14, 2026},
  pages        = {3936--3964},
  publisher    = {{SIAM}},
  year         = {2026},
  url          = {https://doi.org/10.1137/1.9781611978971.144},
  doi          = {10.1137/1.9781611978971.144},
  bibsource    = {dblp computer science bibliography, https://dblp.org}
}

@inproceedings{DBLP:conf/sirocco/BienkowskiKS18,
  author       = {Marcin Bienkowski and
                  Artur Kraska and
                  Pawel Schmidt},
  editor       = {Zvi Lotker and
                  Boaz Patt{-}Shamir},
  title        = {Online Service with Delay on a Line},
  booktitle    = {Structural Information and Communication Complexity - 25th International
                  Colloquium, {SIROCCO} 2018, Ma'ale HaHamisha, Israel, June 18-21,
                  2018, Revised Selected Papers},
  series       = {Lecture Notes in Computer Science},
  pages        = {237--248},
  publisher    = {Springer},
  year         = {2018},
  url          = {https://doi.org/10.1007/978-3-030-01325-7\_22},
  doi          = {10.1007/978-3-030-01325-7\_22},
  bibsource    = {dblp computer science bibliography, https://dblp.org}
}

@article{DBLP:journals/siamcomp/GuptaKP22,
  author       = {Anupam Gupta and
                  Amit Kumar and
                  Debmalya Panigrahi},
  title        = {Caching with Time Windows and Delays},
  journal      = {{SIAM} J. Comput.},
  volume       = {51},
  number       = {4},
  pages        = {975--1017},
  year         = {2022},
  url          = {https://doi.org/10.1137/20m1346286},
  doi          = {10.1137/20M1346286},
  bibsource    = {dblp computer science bibliography, https://dblp.org}
}

@article{DBLP:journals/mst/MariPRS25,
  author       = {Mathieu Mari and
                  Michal Pawlowski and
                  Runtian Ren and
                  Piotr Sankowski},
  title        = {Online Matching with Delays and Stochastic Arrival Times},
  journal      = {Theory Comput. Syst.},
  volume       = {69},
  number       = {1},
  pages        = {12},
  year         = {2025},
  url          = {https://doi.org/10.1007/s00224-024-10207-6},
  doi          = {10.1007/S00224-024-10207-6},
  bibsource    = {dblp computer science bibliography, https://dblp.org}
}

@inproceedings{DBLP:conf/isaac/KrneticM0W20,
  author       = {Predrag Krnetic and
                  Darya Melnyk and
                  Yuyi Wang and
                  Roger Wattenhofer},
  editor       = {Yixin Cao and
                  Siu{-}Wing Cheng and
                  Minming Li},
  title        = {The k-Server Problem with Delays on the Uniform Metric Space},
  booktitle    = {31st International Symposium on Algorithms and Computation, {ISAAC}
                  2020, Hong Kong (Virtual Conference), December 14-18, 2020},
  series       = {LIPIcs},
  pages        = {61:1--61:13},
  publisher    = {Schloss Dagstuhl - Leibniz-Zentrum f{\"{u}}r Informatik},
  year         = {2020},
  url          = {https://doi.org/10.4230/LIPIcs.ISAAC.2020.61},
  doi          = {10.4230/LIPICS.ISAAC.2020.61},
  bibsource    = {dblp computer science bibliography, https://dblp.org}
}

@inproceedings{DBLP:conf/wine/HeLSWWZ23,
  author       = {Kun He and
                  Sizhe Li and
                  Enze Sun and
                  Yuyi Wang and
                  Roger Wattenhofer and
                  Weihao Zhu},
  editor       = {Jugal Garg and
                  Max Klimm and
                  Yuqing Kong},
  title        = {Randomized Algorithm for {MPMD} on Two Sources},
  booktitle    = {Web and Internet Economics - 19th International Conference, {WINE}
                  2023, Shanghai, China, December 4-8, 2023, Proceedings},
  series       = {Lecture Notes in Computer Science},
  pages        = {348--365},
  publisher    = {Springer},
  year         = {2023},
  url          = {https://doi.org/10.1007/978-3-031-48974-7\_20},
  doi          = {10.1007/978-3-031-48974-7\_20},
  bibsource    = {dblp computer science bibliography, https://dblp.org}
}

@article{DBLP:journals/tcs/EmekSW19,
  author       = {Yuval Emek and
                  Yaacov Shapiro and
                  Yuyi Wang},
  title        = {Minimum cost perfect matching with delays for two sources},
  journal      = {Theor. Comput. Sci.},
  volume       = {754},
  pages        = {122--129},
  year         = {2019},
  url          = {https://doi.org/10.1016/j.tcs.2018.07.004},
  doi          = {10.1016/J.TCS.2018.07.004},
  bibsource    = {dblp computer science bibliography, https://dblp.org}
}

@inproceedings{DBLP:conf/focs/00010P21,
  author       = {Anupam Gupta and
                  Amit Kumar and
                  Debmalya Panigrahi},
  title        = {A Hitting Set Relaxation for {\textdollar}k{\textdollar}-Server and
                  an Extension to Time-Windows},
  booktitle    = {62nd {IEEE} Annual Symposium on Foundations of Computer Science, {FOCS}
                  2021, Denver, CO, USA, February 7-10, 2022},
  pages        = {504--515},
  publisher    = {{IEEE}},
  year         = {2021},
  url          = {https://doi.org/10.1109/FOCS52979.2021.00057},
  doi          = {10.1109/FOCS52979.2021.00057},
  bibsource    = {dblp computer science bibliography, https://dblp.org}
}

@inproceedings{GSU26,
      title={Online Matching with Size-Based and Convex Delays}, 
      author={Junhao Gan and Xiao Sun and Seeun William Umboh},
      year={2026},
      booktitle={Approximation, Randomization, and Combinatorial Optimization. Algorithms
                  and Techniques, {APPROX/RANDOM} 2026},
      note={To appear.}
}

@inproceedings{DBLP:conf/icalp/DinitzFU26,
  author       = {Michael Dinitz and
                  Jeremy T. Fineman and
                  Seeun William Umboh},
  editor       = {Sayan Bhattacharya and
                  Danupon Nanongkai and
                  Michael Benedikt and
                  Gabriele Puppis},
  title        = {Learning-Augmented Online Algorithms for Nonclairvoyant Joint Replenishment
                  Problem with Deadlines},
  booktitle    = {53rd International Colloquium on Automata, Languages, and Programming,
                  {ICALP} 2026, Royal Holloway, University of London, Egham, United
                  Kingdom, July 7-10, 2026},
  series       = {LIPIcs},
  volume       = {374},
  pages        = {77:1--77:21},
  publisher    = {Schloss Dagstuhl - Leibniz-Zentrum f{\"{u}}r Informatik},
  year         = {2026},
  url          = {https://doi.org/10.4230/LIPIcs.ICALP.2026.77},
  doi          = {10.4230/LIPICS.ICALP.2026.77},
  bibsource    = {dblp computer science bibliography, https://dblp.org}
}
\bibliographystyle{alpha}

\end{document}